\documentclass{article}
\usepackage{booktabs, multirow, bm}
\usepackage{amsmath, amssymb, amsthm, mathrsfs}
\usepackage{algorithm}
\usepackage{algorithmic}
\usepackage{threeparttable}
\usepackage{indentfirst}
\usepackage{graphicx, tikz, subfigure}
\usepackage{cases}
\usepackage{cite}
\usepackage{tikz}
\usetikzlibrary{positioning}
\usetikzlibrary{arrows.meta}
\usepackage{makecell}
\usepackage{fancyhdr}
\usepackage{array}
\usepackage{geometry}
\usepackage{rotating}
\usepackage{longtable}
\usepackage{caption}
\usepackage{multirow}
\numberwithin{equation}{section}
\usepackage{xcolor}
\usepackage{ulem}
\usepackage{hyperref}
\allowdisplaybreaks[4]
\usepackage{cleveref}
\crefname{section}{Section}{Sections}
\crefname{figure}{Figure}{Figures}
\crefname{table}{Table}{Tables}
\crefname{equation}{}{}
\crefname{theorem}{Theorem}{Theorems}
\crefname{lemma}{Lemma}{Lemmas}
\crefname{remark}{Remark}{Remarks}
\crefname{problem}{Problem}{Subproblems}

\newtheorem{theorem}{Theorem}[section]

\newtheorem{remark}{Remark}[section]

\theoremstyle{definition}

\crefname{ip}{Co-inversion Problem}{ips}
\newtheoremstyle{MyThmStyle}
{}
{}
{}
{}
{\bfseries}
{}
{ }
{\thmname{#1\thmnumber{ #2\hspace{0.5em}}}\thmnote{(#3)}}
\theoremstyle{MyThmStyle}

\crefname{subisp}{inverse source problem}{ips}
\crefname{subiop}{inverse obstacle problem}{ips}
\graphicspath{{figurenew/}}
\definecolor{bananamania}{rgb}{0.98, 0.91, 0.71}

\begin{document}
	
\title{\textbf{An Efficient Machine Learning Framework for Option Pricing via Fourier Transform}}
	\author{
		Liying Zhang\thanks{School of Mathematical Science, China University of Mining and Technology, Beijing 100083, China, lyzhang@lsec.cc.ac.cn },  Ying Gao\thanks{School of Mathematical Science, China University of Mining and Technology, Beijing 100083, China, SQT2300702044@student.cumtb.edu.cn }}
	
	\date{}
	\maketitle

\begin{abstract}	
  The increasing need for rapid recalibration of option pricing models in dynamic markets places stringent computational demands on data generation and valuation algorithms. In this work, we propose a hybrid algorithmic framework that integrates the smooth offset algorithm (SOA) with supervised machine learning models for the fast pricing of multiple path-independent options under exponential Lévy dynamics. Building upon the SOA-generated dataset, we train neural networks, random forests, and gradient boosted decision trees to construct surrogate pricing operators. Extensive numerical experiments demonstrate that, once trained, these surrogates achieve order-of-magnitude acceleration over direct SOA evaluation. Importantly, the proposed framework overcomes key numerical limitations inherent to fast Fourier transform–based methods, including input data consistency and instability in deep out-of-the-money option pricing.
\end{abstract}


\section{Introduction}
The option pricing problem is a core issue in the field of financial mathematics. Efficient and accurate pricing is crucial for the effective operation of financial markets and the realization of investors' risk management goals. The development of option pricing theory relies heavily on mathematical models of stock price dynamics. Three prominent classes of models are particularly noteworthy. First, the Black-Scholes model (\cite{2}) posits that the stock price process $\{S_t\}_{0 \leqslant t\leqslant T}$ follows a geometric Brownian motion (GBM) under the risk-neutral measure $\mathbb{Q}$:
\begin{equation}
\label{eq:BSModelStockSDE}
dS_t = rS_tdt + \sigma S_t dB_t^{\mathbb{Q}},
\end{equation}
where $r$ denotes the risk-free rate, $\sigma$ is the constant volatility, and $B_t^{\mathbb{Q}}$ is a standard Brownian motion under $\mathbb{Q}$. This framework, which assumes continuous sample paths and constant volatility, yields a closed-form solution for European options. Second, the Heston model (HM) (\cite{13}) introduces a stochastic differential equation to characterize the dynamic evolution of volatility, with the specific form as follows:
\begin{equation}
\label{eq:HestonModel}
\left\{
\begin{aligned}
&dS_t = rS_tdt+\sqrt{V_t}S_tdB_t^{(1)\mathbb{Q}}, \\
& dV_t = \kappa(\theta - V_t)dt + \sigma\sqrt{V_t}dB_t^{(2)\mathbb{Q}},
\end{aligned}\right.
\end{equation}
where $V_t$ represents the instantaneous variance, $\kappa$ the mean-reversion rate, and $\theta$ the long-run variance level. While maintaining the assumption of continuous price trajectories, this model allows volatility to evolve stochastically, thereby accounting for leptokurtic asset returns. Third, exponential Lévy models accommodate discontinuous movements in stock prices through the inclusion of jump components (\cite{16}). This category encompasses specifications such as jump-diffusion processes (\cite{20}), the variance Gamma process (\cite{19}), and the Carr–Geman–Madan–Yor (CGMY) process
(\cite{6}). By relaxing the path continuity assumption, exponential Lévy models are particularly adept at capturing the abrupt, large-scale movements frequently observed in financial markets.

Due to the complex mathematical operations in the aforementioned stock price models, closed-form solutions are often hard to derive. Consequently, numerical algorithms have become a key method to price both single and multiple options. For single option pricing, beyond the classical Monte Carlo method(\cite{1,22}), and then various numerical pricing techniques have been developed in the financial field. An algorithm based on the binomial tree model introduced by Cox et al. (\cite{8}) employs discrete-time approximation of asset price paths, establishing a flexible pricing framework for path-dependent derivatives such as American options. Boyle (\cite{3}) subsequently developed an algorithm based on the trinomial tree model, incorporating an additional intermediate price movement to enhance numerical accuracy and convergence rate, with extensions to multi-asset options. Brennan et al. (\cite{4}) demonstrate that the option pricing problem can be formulated as a partial differential equation (PDE) boundary value problem and can be effectively solved using the finite difference method. Meanwhile, this method is also useful for some path-dependent derivatives, such as American options, where early exercise
features transform the pricing task into a free boundary problem (\cite{24}).

Although these methods offer distinct advantages in terms of flexibility, their application remains constrained, particularly in high-dimensional or high-precision pricing scenarios where computational costs increase substantially. Consequently, the development of efficient and robust pricing algorithms represents an ongoing research imperative. \cite{7} proposes an algorithm based on Fourier transform (FT), which we refer to as the Carr-Madan algorithm (CMA). The CMA introduces an offset term to modify the option pricing formula, ensuring that its Fourier transform admits a closed-form expression, and subsequently recovers the option price via an inverse transform. To further improve efficiency, we leverage the relationship between the smoothness of a function and the decay rate of its Fourier transform tail to propose a novel offset term. This modification enhances the algorithm's efficiency, and the corresponding improved method is termed the smooth offset algorithm (SOA).
Numerical experiments show that under two option types (European options and digital options) and three stock price models (GBM, HM, and EVGP), when error tolerances are equivalent, the operational efficiency of SOA is significantly superior to that of CMA, with its running time only accounting for $ 60\%-70\% $ of CMA's. Details are provided in Section \ref{sec:TheEffcientPricingAlgorithm}.

 The research on algorithms for single option pricing is relatively mature. However, with the continuous expansion of business scales in financial institutions such as investment banks, the demand for pricing large-scale option portfolios has become increasingly prominent, and the traditional single option pricing algorithms can no longer meet the efficiency requirements of batch processing. Existing studies have also begun to focus on pricing algorithms for multiple options. For example, under the discrete-time framework, Derman and Kani proposed a pricing algorithm based on the implied binomial tree (IBT) model (\cite{9}). This algorithm can derive the local volatility surface from observed market option prices, providing a flexible numerical tool for the fitting and pricing of multiple option portfolios. However, it relies on local volatility theory, requires the assumption of a unique volatility surface, and imposes strict constraints on input option data. The Brodie-Kaya (BK) algorithm proposed by Brodie and Kaya (\cite{5}) constructs accurate simulation paths for stochastic volatility models, enabling batch pricing of options with various payoff structures through a single simulation and significantly improving the computational efficiency of the Monte Carlo method. Nevertheless, this algorithm is only applicable to affine diffusion processes and is difficult to extend to non-affine models such as exponential Lévy processes. 

Compared with traditional algorithms, machine learning (ML) has seen growing adoption in option pricing.
Recent research increasingly leverages ML models to establish the mapping between option attributes and their theoretical prices. For example, \cite{15} investigates ordinary least squares, radial basis function networks, multi-layer perceptron networks, and projection pursuit on the Black–Scholes option pricing problem. It demonstrates that ML models achieve accurate pricing and delta-hedging. Likewise, \cite{11} analyzes basket call options. It shows that neural networks (NNs) can compute prices at a dramatically faster speed. This speed is approximately one million times faster than traditional methods, while maintaining satisfactory accuracy. 
To address the problems of inadequate model accuracy and anti-noise stability in option pricing, \cite{17} adopts an experimental evaluation strategy integrated with financial theories, investigates the application effect of ensemble learning and the interaction between sliding windows and noise, verifies the applicability of ensemble learning to option pricing, and clarifies the adaptation logic between the two as well as the influence law on the pricing stability of the model.
Other ML models address option pricing problems by numerically solving the associated PDEs through NNs. This method is commonly referred to as physics informed neural networks (PINNs) (\cite{10,12,21,26}). Although traditional machine learning models have achieved certain breakthroughs in pricing accuracy and efficiency, they still face core challenges in practical financial applications. The dynamic changes in market conditions require models to have rapid updating capabilities, while pricing models typically rely on large-scale data for training. Traditional numerical methods exhibit inefficiencies in data generation, and illiquid option contracts also suffer from data scarcity, making it difficult to support high-frequency iterations. Thus, we propose a pricing algorithm integrating SOA with machine learning. This algorithm uses option data generated by SOA as the training set, allowing the data-driven model to learn the mapping relationship between option attributes and their corresponding prices. Numerical experiments demonstrate  that compared with fast Fourier transform (FFT) the machine learning models trained in this way are not constrained by input conditions and maintains good numerical stability.
 
The remainder of this paper is organized as follows. Section \ref{Preliminary} presents the stock price models. 
Section \ref{sec:TheEffcientPricingAlgorithm} briefly reviews CMA and introduces the theoretical foundations of the SOA method, with its advantages being verified through numerical experiments. The FFT-based
and ML-based frameworks are discussed in Section \ref{sec:FrameworksForPricingMultOptions} and \ref{sec:MLFrameworksForPricingMultOptions}, with numerical experiments that compare them and the plain sequential pricing approach by SOA. Finally, Section \ref{sec:Conclusions} concludes.

\section{Stock Price Models\label{Preliminary}}
\subsection{Unified Framework of Stock Price Models}
This section incorporates three types of stock price models, namely geometric Brownian motion (GBM), Heston model (HM), and exponential variance Gamma process (EVGP), into a unified framework and derives the corresponding characteristic function for each model. Consider a stock price process $ \{S_t\}_{0 \leqslant t\leqslant T} $ over the time interval $ [0,T] $. Let $ \{X_t\}_{0 \leqslant t\leqslant T} $ be a Lévy process under measure $ \mathbb{Q} $, and the exponential Lévy process for the stock price is driven by the following stochastic differential equation (SDE):
\begin{equation*}
\label{eq:ExpLevyProcessSDE}
dS_t = r S_t \, dt + S_t \, dX_t, \qquad 0 \leqslant t \leqslant T,
\end{equation*}
where $ r $ denotes the risk-free interest rate. The solution to this SDE can be expressed as
\begin{equation}
\label{eq:ExpLevyProcessSDESolution}
S_t = S_0 \exp \big[ (r + \zeta)t + X_t \big]=
S_0 \exp(rt + X_t^{\dagger})
, \qquad 0 \leqslant t \leqslant T,
\end{equation}
where $X_t^{\dagger} = \zeta t + X_t$, and $\zeta$ serves as a compensator, ensuring that the discounted process $\left\{e^{-rT}S_t\right\}_{0 \leqslant t\leqslant T}$ is a martingale under measure $\mathbb{Q}$. 
To integrate the three stock price models into a unified framework, we rewrite each of them in the form of (2.1) and derive the characteristic function of $ X_t^\dagger $ using the characteristic function of $ X_t $, specifically $  \Phi_{X_t^\dagger}(z) = e^{i\zeta z t} \Phi_{X_t}(z).   $

For the GBM, its compensator term is $ \zeta = -\frac{1}{2}\sigma^2 $, where $ X_t = \sigma B_t^\mathbb{Q} $ with the characteristic function
$$  \Phi_{X_t}(z) = \exp\left(-\frac{1}{2}\sigma^2 z^2 t\right).   $$

For the EVGP, its compensator term is given by
$ \zeta = \frac{1}{\nu} \ln\left(1 - \theta\nu - \frac{1}{2}\sigma^2\nu\right),  $ 
and the corresponding Lévy process is $ X_t = \text{VGP}_t^{(\theta,\sigma,\nu)} $ with the characteristic function
$$  \Phi_{\text{VGP}_t^{(\theta,\sigma,\nu)}}(z) = \left(1 - iz\theta\nu + \frac{1}{2}\sigma^2\nu z^2\right)^{-\frac{t}{\nu}}.  $$ 

Especially, the Heston model is not an exponential Lévy process, and as a result, no explicit expression for$  X_t^\dagger $ can be obtained. However, the characteristic function of $X^{\dagger}_t$ can still be obtained by exploiting the log-price representation
\begin{equation*}
\ln S_t = \ln S_0 + rt + X_t^{\dagger}.
\end{equation*}
From this relation, the characteristic function of $X_t^{\dagger}$ can be derived
\begin{equation*}
\Phi_{X_t^{\dagger}}(z) = \mathbb{E}\left[e^{iz(\ln S_t - \ln S_0 - rt)}\right] = e^{-iz(\ln S_0 + rt)}\Phi_{\ln S_t}(z).
\end{equation*}
The closed-form expression of $\Phi_{\ln S_t}(z)$ is available (\cite{14})
, and is given by
\begin{equation*}
\begin{aligned}
\Phi_{\ln S_t}(z) = \frac{\exp \left[ i z \ln S_0+i zr t+\frac{\kappa \theta t(\kappa-i \rho \sigma z)}{\sigma^2}\right]}{\left(\cosh \frac{\tau t}{2}+\frac{\kappa-i \rho \sigma z}{\tau} \sinh \frac{\tau t}{2}\right)^{\frac{2 \kappa \theta}{\sigma^2}}} 
\exp \left[\frac{-\left(z^2+i z\right) V_0}{\tau \operatorname{coth} \frac{\tau t}{2}+\kappa-i \rho \sigma z}\right],
\end{aligned}
\end{equation*}
where $\tau=\sqrt{\sigma^2\left(z^2+iz\right)+(\kappa-i \rho \sigma z)^2}$, and $V_0$ is the initial variance.
\begin{remark}
To simplify notation, the initial stock price  $ S_0 $  is normalized to 1, and results for any arbitrary  $ S_0 $  can be obtained by rescaling the normalized price. For European call options: 
\begin{equation*}
e^{-rT}\mathbb{E}^Q\left[(S_T - K)^+\right] = S_0 e^{-rT}\mathbb{E}^Q\left[\left(\exp\left(rT + X_T^\dagger\right) - \frac{K}{S_0}\right)^+\right] ,
\end{equation*}
 where  $ e^{-rT}\mathbb{E}^Q\left[\left(\exp\left(rT + X_T^\dagger\right) - \frac{K}{S_0}\right)^+\right] $  is the normalized price under $  S_0' = 1 $  and  $ K' = \frac{K}{S_0} $ , and the true price needs to be multiplied by $  S_0  $. Similarly, for digital options: 
\begin{equation*}
  e^{-rT}\mathbb{E}^Q\left[1_{\{S_T \geqslant K\}}\right] = e^{-rT}\mathbb{E}^Q\left[1_{\left\{\exp\left(rT + X_T^\dagger\right) \geqslant \frac{K}{S_0}\right\}}\right] ,
\end{equation*}  
 whose scaling factor is 1, meaning the normalized price equals the true price.
\end{remark}
\section{SOA for Single Option Pricing \label{sec:TheEffcientPricingAlgorithm}}
\subsection{Revised CMA}
CMA modifies the original option pricing formula by introducing an offset term, ensuring that the FT of
the resulting expression exists and admits a closed-form representation. Consider a European or digital option whose underlying stock price $\left\{S_t\right\}_{0\leqslant t\leqslant T}$ follows exponential L{\'e}vy process  with $S_t = \exp\left(rt + X_t^\dagger\right)$. The price function modified by the CM offset term is
	\begin{equation}
	\label{eq:CMPriceFuncWithOffset}
	\hat{V}(k) = \begin{cases}
	V(k) - e^{-rT}\left(e^{rT}  - e^k\right)^+ , \text{\,\,\,\,for European options},\\
	V(k) - e^{-rT}\mathrm{1}_{e^{rT}\geqslant e^k}, \text{\,\,\,\,for digital options}.
	\end{cases}
	\end{equation}
	Then the FT of $\hat{V}$ is given by
	\begin{equation*}
	\label{eq:EtaFuncCMOffset}
	\eta(z)=\mathcal{F}\left[{\hat{V}}\right](z) = 
	\begin{cases}
	\frac{1}{iz(iz+1)}e^{izrT}\left[\Phi^{\dagger}(z-i)-1\right], \text{\,\,\,\,for European options},\\
	\frac{1}{iz}e^{(iz-1)rT}\left[\Phi^{\dagger}(z)-1\right], \text{\,\,\,\,for digital options},
	\end{cases}
	\end{equation*}
	where $\Phi^\dagger$ is the FT of $X^\dagger_T$.

The function $\hat{V}(k)$ can be obtained via the inverse FT of $\eta$:
\begin{equation}
\label{eq:InverseFTOfEtaFunc}
\hat{V}(k) = \mathcal{F}^{-1}\left[\eta\right](k) =  \frac{1}{2\pi}\int_{-\infty}^{\infty}e^{-izk}\eta(z)dz.
\end{equation} 
The integral in \eqref{eq:InverseFTOfEtaFunc} can only be evaluated numerically. Once $\hat{V}(k)$ has been obtained, $V(k)$ is given by
\begin{equation*}
V(k) = \begin{cases}
\hat{V}(k) + e^{-rT}\left(e^{rT}  - e^k\right)^+ , \text{\,\,\,\,for European options},\\
\hat{V}(k) + e^{-rT}\mathrm{1}_{e^{rT}\geqslant e^k}, \text{\,\,\,\,for digital options}.
\end{cases}
\end{equation*}

To develop an algorithm that achieves better computational efficiency than CMA, it is essential to ensure that the tail of the function $\eta$ decays more rapidly. A faster tail decay rate accelerates the convergence of the numerical integration used to evaluate \eqref{eq:InverseFTOfEtaFunc}, thereby improving the overall efficiency.
The intuition behind this point is illustrated in Figure \ref{fig:ImpactsOfTailDecayOnNumInteg}, where two functions  $\eta_i,$ $i=1,2$ are considered. The objective is to evaluate integrals 
\begin{equation*}
\int_{0}^{\infty}\eta_i(z)dz, i=1,2,
\end{equation*}
which are approximated using a truncation point $B>0$:
\begin{equation*}
\int_{0}^{\infty}\eta_i(z)dz\approx\int_{0}^{B}\eta_i(z)dz, i=1,2.
\end{equation*}
\begin{figure}[t]
	\centering
	\caption{Impacts of Tail Decay Rates on the Convergence Speed of Numerical Integrations}
	\label{fig:ImpactsOfTailDecayOnNumInteg}
	\includegraphics[scale=0.7]{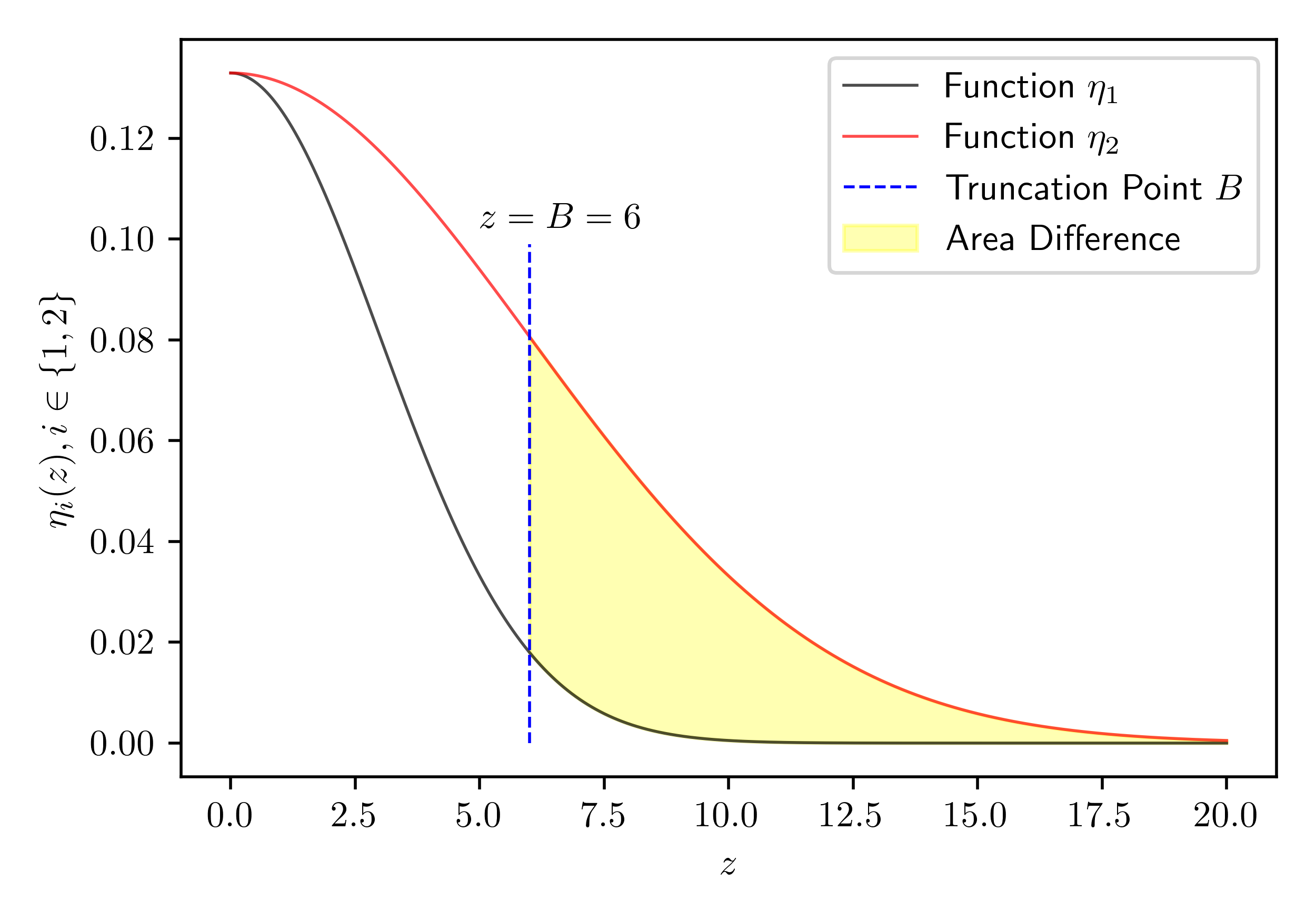}
\end{figure}
Evidently, this truncation introduces a larger approximation error for $\eta_2$ since the tail of $\eta_2$ decays more slowly. The vertical truncation line at $z=B$ omits a non-negligible portion of the area beyond $z=B$. In contrast, $\eta_1$ exhibits a faster decay rate and approaches zero in the neighborhood of $z=B$, leading to a more accurate approximation. Consequently, to achieve comparable accuracy for $\eta_2$, a larger truncation point $B$ must be selected. However, expanding the integration interval $[0,B]$ requires a greater number of subintervals in the numerical integration, thereby increasing the computational cost. In the application of this paper, $\eta$ is the FT of $\hat{V}$, and thus the properties of $\hat{V}$ must be closely linked to the tail decay rate of $\eta$. Specifically, a smoother $\hat{V}$ results in a faster decay of $\eta$ at its tail. See Theorem 2.5 in Reference \cite{25} for details.

The CM offset term is defined as the discounted payoff evaluated at the terminal stock price, i.e. $\mathbb{E}^{\mathbb{Q}}[S_T] = e^{rT}$. However, because the payoff function itself lacks smoothness, the resulting CM offset term also exhibits poor smoothness, which negatively affects convergence. 
To address this issue, the smooth offset term in SOA is constructed by averaging the payoff values over a range of stock prices rather than evaluating it at a single point. This averaging operation improves the smoothness of the offset term and thereby accelerates the algorithm. The formal result is presented in Theorem \ref{thmSOMAResults}.
\begin{theorem}
	\label{thmSOMAResults}
	Consider a European or digital option whose underlying stock price $\left\{S_t\right\}_{0\leqslant t\leqslant T}$ follows exponential L{\'e}vy process  with $S_t = \exp\left(rt + X_t^\dagger\right)$. The price function modified by the smooth offset term is
	\begin{equation*}
	\label{eq:SOPriceFuncWithOffset}
	\hat{V}(k) = \begin{cases}
	V(k) - e^{-rT}\mathbb{E}^{\mathbb{Q}}\left[(S_T^\wedge-e^k)^+\right], \text{\,\,\,\,for European options},\\
	V(k) - e^{-rT}\mathbb{E}^{\mathbb{Q}}\left[ \mathbf{1}_{S_T^\wedge\geqslant e^k} \right], \text{\,\,\,\,for digital options},
	\end{cases}
	\end{equation*}
	where $S_T^\wedge$ is the terminal value of process $\left\{S_t^{\wedge}\right\}_{0\leqslant t\leqslant T}$ that follows SDE
	\begin{equation*}
	\begin{cases}
	dS_t^\wedge = rS_t^\wedge dt  + S_t^\wedge  dB_t^\mathbb{Q},\\
	S_0^\wedge = 1,
	\end{cases}
	\end{equation*}
	which is a GBM under measure $\mathbb{Q}$ with volatility 100\%.
	Then the FT of $\hat{V}$ is given by
	\begin{equation*}
	\label{eq:EtaFuncCMOffset}
	\eta(z)=\mathcal{F}\left[{\hat{V}}\right](z) = 
	\begin{cases}
	\frac{1}{iz(iz+1)}e^{izrT}\left[\Phi^{\dagger}(z-i)-\Phi^{\wedge}(z-i)\right], \text{\,\,\,\, for European options},\\
	\frac{1}{iz}e^{(iz-1)rT}\left[\Phi^{\dagger}(z)-\Phi^{\wedge}(z)\right], \text{\,\,\,\,for digital options},
	\end{cases}
	\end{equation*}
	where $\Phi^\dagger$ is the FT of $X^\dagger_T$ and function $\Phi^\wedge$ is the FT of random variable $-\frac{T}{2}+B_T^{\mathbb{Q}}$, given by
	\begin{equation*}
	\Phi^\wedge(z) = \exp\left[-\frac{T}{2}(iz+z^2)\right].
	\end{equation*}
\end{theorem}
\begin{proof}
The smooth offset term 
$ e^{-rT}\mathbb{E}^{\mathbb{Q}}\left[(S_T^\wedge-e^k)^+\right]$ 
can be interpreted as the price of a virtual European option whose underlying stock price follows the process  $\left\{S_t^\wedge\right\}_{0\leqslant t\leqslant T}$. 
Since $\left\{S_t^\wedge\right\}_{0\leqslant t\leqslant T}$ follows GBM, which is a special case of exponential L{\'e}vy processes. 

Let $V^\wedge(k) = e^{-rT}\mathbb{E}^{\mathbb{Q}}\left[(S_T^\wedge-e^k)^+\right]$. Then $V^\wedge(k)$ denotes the price of the virtual European option mentioned above. It is worth noting that although $\left\{S_t\right\}_{0\leqslant t\leqslant T}$ and $\left\{S_t^\wedge\right\}_{0\leqslant t\leqslant T}$ are distinct processes, they share the same terminal expectation as the discounted value of each process must constitute a $\mathbb{Q}-$martingale:
\begin{equation*}
\mathbb{E}^{\mathbb{Q}}\left[S_T\right] =  	\mathbb{E}^{\mathbb{Q}}\left[S_T^{\wedge}\right]  = S_0e^{rT} = e^{rT}.
\end{equation*}
The modified price function can thus be expressed as
\begin{equation*}
\begin{aligned}
\hat{V}(k) =& V(k)  - V^\wedge(k) 
=  \left[V(k) - e^{-rT}\left(\mathbb{E}^{\mathbb{Q}}\left[S_T\right]-e^k\right)^+\right]  -  \left[V^\wedge(k) - e^{-rT}\left(\mathbb{E}^{\mathbb{Q}}\left[S_T^{\wedge}\right]-e^k\right)^+\right].
\end{aligned}
\end{equation*}
Define
\begin{equation*}
V^A(k) = V(k) - e^{-rT}\left(\mathbb{E}^{\mathbb{Q}}\left[S_T\right]-e^k\right)^+ \text{\,\,and\,\,}V^B(k) =V^\wedge(k)-e^{-rT}\left(\mathbb{E}^{\mathbb{Q}}\left[S_T^{\wedge}\right]-e^k\right)^+,
\end{equation*}
and we obtain
\begin{equation*}
\mathcal{F}\left[V^A\right](z) =  \frac{1}{iz(iz+1)}e^{izrT}\left[\Phi^{\dagger}(z-i)-1\right]
\end{equation*}
and
\begin{equation*}
\mathcal{F}\left[V^B\right](z) =  \frac{1}{iz(iz+1)}e^{izrT}\left[\Phi^{\wedge}(z-i)-1\right].
\end{equation*}
Therefore, the FT of $\hat{V}$ is
\begin{equation*}
\eta(z) = \mathcal{F}\left[V^A\right](z) - \mathcal{F}\left[V^B\right](z) = 
\frac{1}{iz(iz+1)}e^{izrT}
\left[\Phi^{\dagger}(z-i) -  \Phi^{\wedge}(z-i)\right].
\end{equation*}
The proof is provided for European options, while digital options are treated analogously. 
\end{proof}
The function $\hat{V}(k)$ is obtained by applying the inverse FT to $\eta(z)$, i.e. $\hat{V}(k) = \mathcal{F}^{-1}\left[\eta\right](k)$, i.e.
\begin{equation*}
V(k) = \begin{cases}
\hat{V}(k) + e^{-rT}\mathbb{E}^{\mathbb{Q}}\left[\left(S_T^\wedge - e^k\right)^+\right], \text{\,\,\,\,for European options},\\
\hat{V}(k) + e^{-rT}\mathbb{E}^{\mathbb{Q}}\left[\mathbf{1}_{S_T^\wedge \geqslant e^k}\right], \text{\,\,\,\,for digital options},
\end{cases}
\end{equation*}
where the smooth offset terms admit close-form expressions. Specifically, let $F_N$ be the cumulative density function of the standard normal distribution, and then we have
\begin{equation*}
\begin{cases}
e^{-rT}\mathbb{E}^{\mathbb{Q}}\left[\left(S_T^\wedge - e^k\right)^+\right] = F_N(d_1) - e^{k-rT}F_N(d_2),\\
e^{-rT}\mathbb{E}^{\mathbb{Q}}\left[\mathbf{1}_{S_T^\wedge \geqslant e^k}\right] = e^{-rT}F_N(d_2), \\
d_{1,2} = \frac{1}{\sqrt{T}} \left[-k + \left(r\pm \frac{1}{2}\right)T\right].
\end{cases}
\end{equation*}

To illustrate the contrast between the two offset terms, we draw offset terms for both option types in Figure \ref{fig:CompareOffsets}.
It can be observed that the smooth offset terms eliminate the sharp corners at $k = rT$, resulting in continuous and differentiable curves.
\begin{figure}[h]
	\centering
	\caption{CM and Smooth Offset Terms for European and digital options ($ r $=2.0\%, $ T $=3 Months)}
	\label{fig:CompareOffsets}
	\includegraphics[scale=0.62]{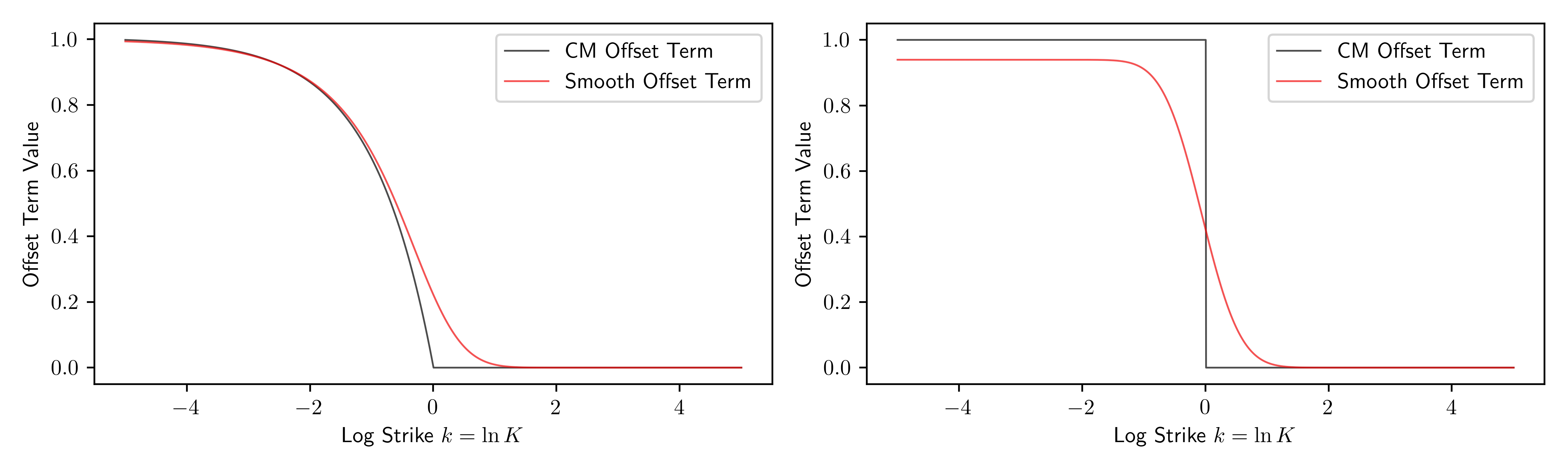}
\end{figure}
Additionally, Figure \ref{fig:CompareEtaFunctions} compares the functions $\eta$ corresponding to different combinations of offset terms and option types under the settings given by Table \ref{tab:DetailsOfOptions} and \ref{tab:DetailsOfStockPriceModels} presented in section \ref{subsec:SOANumExp}.
These observations indicate that the $\eta$ functions derived from the smooth offset terms decay to $0$ significantly faster than those based on the CM offset terms.
This advantage is particularly pronounced for digital options, whose payoff functions exhibit lower smoothness than European options.
\begin{figure}[ht]
	\centering
	\caption{Cases of $\eta$ under Different Settings for the Two Offset Terms}
	\label{fig:CompareEtaFunctions}
	\includegraphics[scale=0.52]{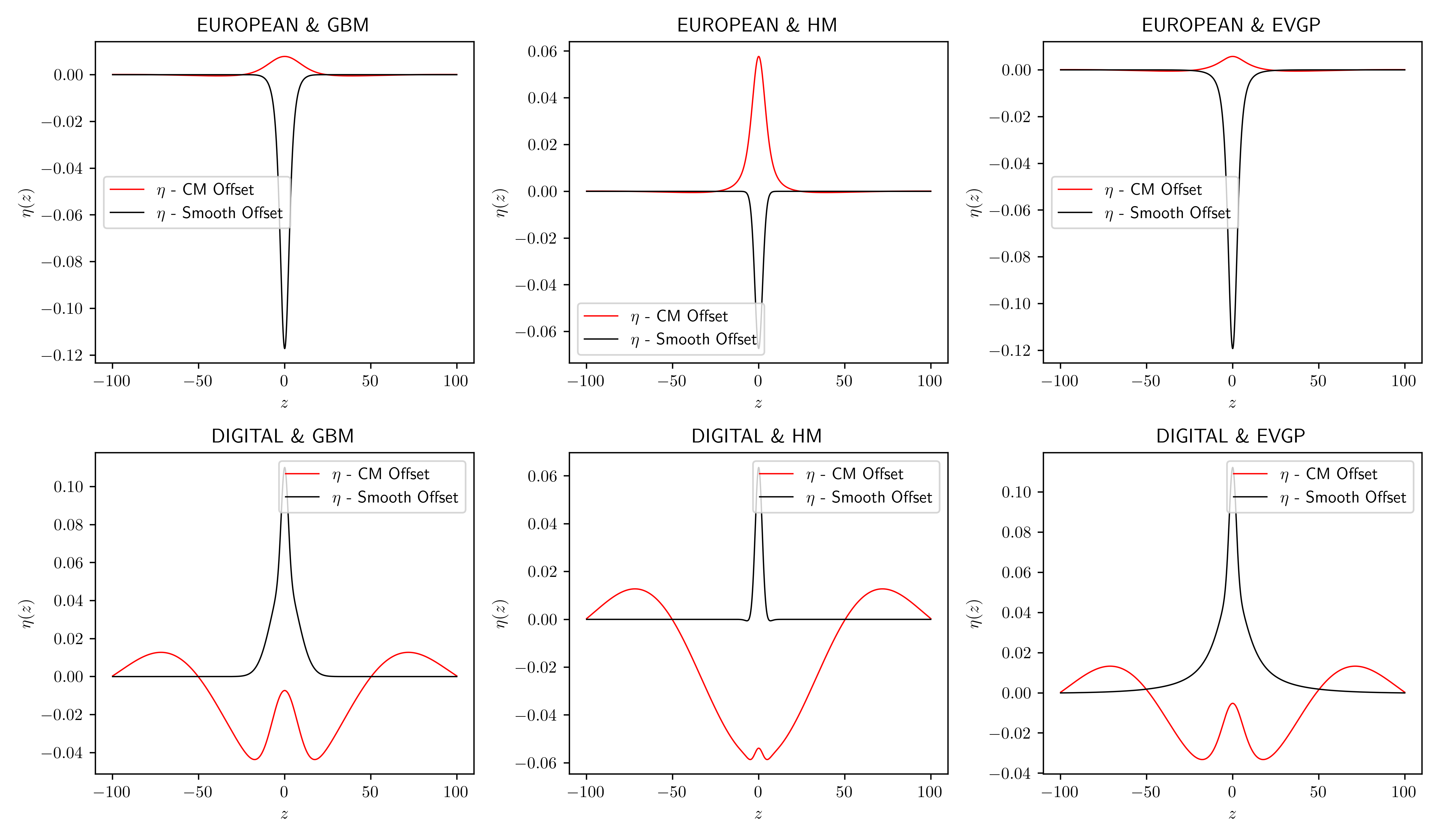}
\end{figure}

\subsection{Numerical Experiments for SOA and CMA \label{subsec:SOANumExp}}

We verify the advantages of SOA over CMA through numerical experiments. The tests cover two option types, including European and digital, as well as three stock price models, including GBM, HM and EVGP.
The specifications of the options and stock price models are summarized in Table \ref{tab:DetailsOfOptions} and \ref{tab:DetailsOfStockPriceModels}, respectively. To verify the accuracy of CMA and SOA, appropriate computational benchmarks are required. When stock prices are assumed to follow GBM, closed-form solutions are available for both European and digital options, and these solutions serve as the benchmarks for this case. In contrast, for HM and EVGP, no closed-form expressions exist. Consequently, for these cases, MC methods are employed to generate benchmarks.
\begin{table}[h]
	\renewcommand{\arraystretch}{1.5} 
	\centering
	\caption{Details of Options in the Numerical Experiment}
	\label{tab:DetailsOfOptions}
	\begin{tabular}{cccccccc}
		\hline
	   Option Type & \makecell{Spot Price\\($S_0$)} & \makecell{Strike Price\\($K$)} & \makecell{Maturity\\($T$)} & \makecell{Risk-free Rate\\($r$)} \\
		\hline
		European  & $150.0$ & $100.0$ & $3$ months ($0.25$) & $2.0\%$  \\
		Digital & $150.0$ & $100.0$ & $3$ months ($0.25$) & $2.0\%$  \\
		\hline
	\end{tabular}
\end{table}
\begin{table}
	\renewcommand{\arraystretch}{1.5} 
	\centering
	\caption{Details of Stock Price Models in the Numerical Experiment}
	\label{tab:DetailsOfStockPriceModels}
	\begin{tabular}{cccccccc}
		\hline
		Model Type&&&& Parameters  \\
		\hline
		GBM  &&&& $\sigma = 0.25$  \\
		HM   &&&& $(\kappa, \theta, \sigma, \rho, V_0) = (2.30, 0.36, 0.10, 0.60, 0.49)$   \\
		EVGP &&&& $(\theta, \sigma, \nu) = (0.10, 0.20, 0.30)$   \\
		\hline
	\end{tabular}
\end{table}

Hereafter, we elaborate on the specific implementation processes of CMA and SOA. The most important step of CMA or SOA is numerically evaluating the following integral:
\begin{equation}
\label{eq:FourierInvTransIntergralWithTruncationPoint}
\hat{V}(k)=\frac{1}{2\pi}\int_{-\infty}^{\infty}e^{-izk}\eta(z)dz
=\frac{1}{\pi}\int_0^\infty e^{-izk}\eta(z)dz.
\end{equation}
The integral in  \eqref{eq:FourierInvTransIntergralWithTruncationPoint} is numerically computed using the Simpson's rule: given a truncation point  $ B $ , the interval $ [0, B] $ is divided into  $ N $  subintervals, and then:
\begin{equation}
\label{eq:NumericalFormulaHatVk}
\hat{V}(k) \approx \frac{\Delta z}{\pi}\sum_{j=0}^{N}w_je^{-ijk\Delta z}\eta(j\Delta z),
\end{equation}
where the weights $w_j$ are the weight coefficients corresponding to the node in Simpson's numerical integration rule.

Equation \eqref{eq:NumericalFormulaHatVk} better illustrates the advantages of SOA: the smoothness of $\hat{V}$ leads to rapid tail decay of $\eta$ as $z \to \infty$, thereby allowing for a smaller truncation point $B$ to achieve relatively accurate option pricing. As shown in Figure \ref{fig:CompareIntegralLimits}, which plots the option price against $B$ for different cases, both CMA and SOA converge to the reference value. However, SOA attains convergence with a significantly smaller $B$, while CMA exhibits slower convergence$-$a difference that is particularly pronounced for digital options. Since SOA achieves satisfactory accuracy with a smaller $B$, the number of subintervals $N$ required for the numerical integration in \eqref{eq:NumericalFormulaHatVk} can be reduced, thus improving computational efficiency.
\begin{figure}[t]
	\centering
		\caption{Relationship between Truncation Point $ B $ and Option Prices}
	\label{fig:CompareIntegralLimits}
	\includegraphics[scale=0.52]{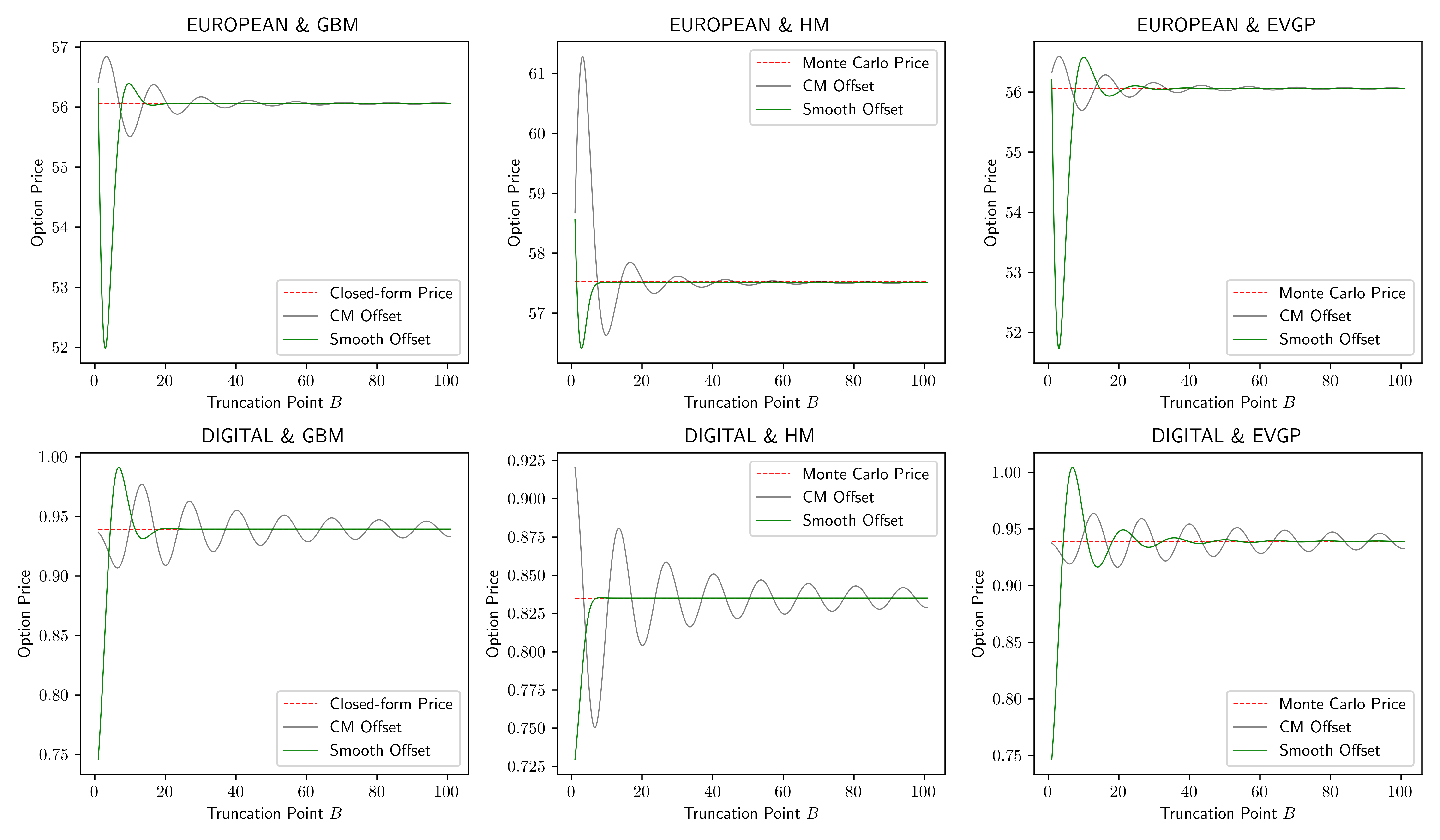}
\end{figure}
In practical applications, both CMA and SOA involve two tunable parameters: the truncation point ($B$) and the number of subintervals in the numerical integration ($N$). It is crucial to choose these parameters appropriately so that the algorithm can strike a balance between the accuracy and the efficiency. A value of $B$ that is too small leads to intolerable errors, whereas an excessively large value of $B$ increases computational cost. The choice of $N$ is inherently dependent on $B$: as $B$ increases, a larger $N$ is required to maintain the desired level of accuracy. 

To ensure a fair comparison of execution times between CMA and SOA, a systematic policy for determining $B$ and $N$ is required. When identical values of $B$ and $N$ are used for both algorithms, their execution times may appear similar. However, such a comparison is misleading, as it disregards the differences in accuracy. In particular, the error associated with CMA can be substantially larger than that of SOA if the selected $B$ and $N$ are insufficient for CMA to achieve convergence.

To address this issue, an algorithm is developed to determine the optimal parameter settings for CMA and SOA. The execution times of both algorithms are then compared under their respective optimal configurations, ensuring a fair performance evaluation. For the $6$ cases, let $\left\{V^{(n,\ell)}\right\}_{\ell=1}^{6}$ be the option prices computed numerically by either CMA or SOA, and $\left\{V^{(b, \ell)}\right\}_{i=\ell}^{\ell}$ represent the corresponding benchmark prices.
The relative pricing error for the $\ell^{\text{th}}$ case is defined as
\begin{equation*}
e_\ell = \left\lvert\frac{V^{(n, \ell)}}{V^{(b, \ell)}} - 1 \right\rvert, \ell=1,2,...,6,
\end{equation*}
and the mean relative pricing error across all cases is given by
\begin{equation*}
\bar{e} = \frac{1}{6}\sum_{\ell=1}^{6}e_\ell,
\end{equation*}
which measures the overall accuracy of both algorithms.

The proposed search algorithm is based on a grid-search strategy.
A predefined error threshold, denoted by $e_{\text{th}}$, is specified, and the mean relative pricing error $\bar{e}$ is evaluated over a range of $(B, N)$ pairs.
The optimal parameter pair is identified as the smallest $(B, N)$ combination that satisfies the condition $\bar{e} \leqslant e_{\text{th}}$.

The settings of the search algorithm are specified as follows: $B_{\min} = 10, B_{\max} = 2000, \delta B = 10, \iota_{\min} =0.1, \iota_{\max} = 5.0, \delta \iota=0.1$. The error threshold is set to $e_{\text{th}}=2 \,\,\text{basis points (bps)}$. 
For prudential considerations, financial institutions typically require that the pricing discrepancy between different algorithms for the same derivative contract should not exceed $20$$-30$ bps, depending on the contract type.
Hence, the selected threshold $e_{\text{th}}$ in this paper is intentionally more stringent than those adopted in practice. The optimal parameters obtained from the search algorithm are as follows:
\begin{itemize}
	\item $B^*=360, \iota^*=1.6, N^*=\lfloor\iota^*B^*\rfloor=576$ for CMA (with mean error $\bar{e} = 1.9716$ bps);
	\item $B^*=40, \iota^* = 1.6, N^*=\lfloor\iota^*B^*\rfloor=64$ for SOA (with mean error $\bar{e} = 1.9789$ bps).
\end{itemize}
After obtaining the optimal parameters, to further quantify the computational efficiency of SOA relative to CMA, we estimate the ratio $\frac{t_{\text{SOA}}}{t_{\text{CMA}}}$ using the linear regression model without an intercept term \eqref{eq:OLSEqnForExeTimes}:
\begin{equation}
\label{eq:OLSEqnForExeTimes}
t_{\rm SOA} =\beta t_{\rm CMA} + \varepsilon.
\end{equation} 
The specific results are shown in Figure \ref{fig:OLSExeTime} and Table \ref{tab:OLSExeTime}. 
The data exhibit a statistically significant linear relationship, with an estimated slope of $\hat{\beta} = 0.2259$, indicating that the computational time required by SOA is, on average, approximately $22.59\%$ of that required by CMA.

\begin{figure}[h]
	\centering
	\caption{Relationship between the Execution Times of SOA and CMA}
	\label{fig:OLSExeTime}
	\includegraphics[scale=0.6]{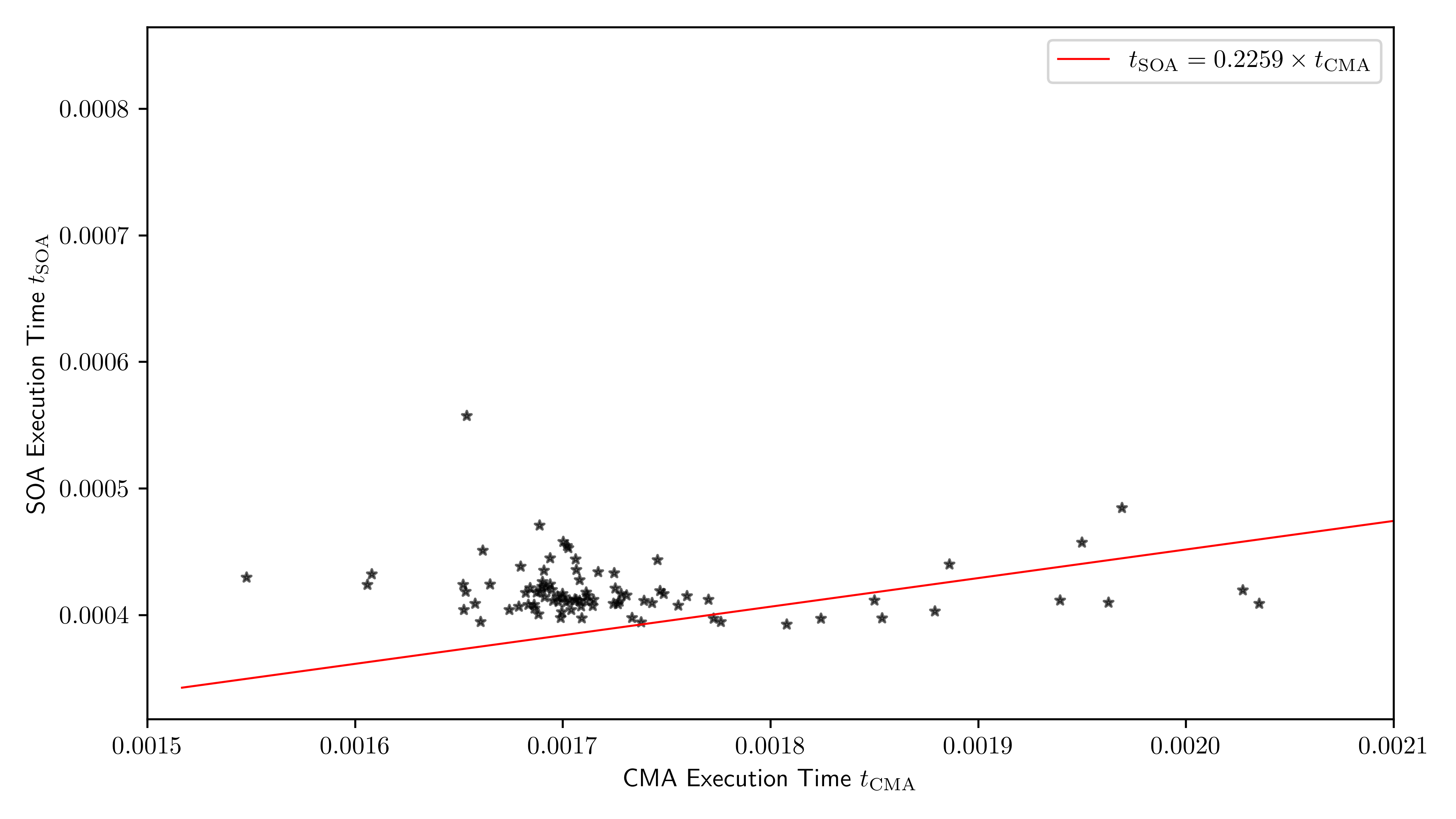}
\end{figure}

\begin{table}[h]
	\renewcommand{\arraystretch}{1.3}
	\centering
	\begin{threeparttable}
		\caption{Regression Results for \eqref{eq:OLSEqnForExeTimes}}
		\label{tab:OLSExeTime}
		\begin{tabular}{ccccccccccccccccc}
			\hline
			Coefficient & Estimate & Standard Error & $t-$Statistic & $p-$Value & $95\%$ Confidence Interval \\
			\hline
			$\beta$ & $0.2259$ & $1.05\times 10^{-3}$&$236.450$ & $4.4587\times 10^{-138}$& $[0.2485, 0.2527]$\\
			\hline
		\end{tabular}
	\end{threeparttable}
\end{table}

\section{FFT for Multiple Options Pricing
	\label{sec:FrameworksForPricingMultOptions}}
\subsection{FFT-based Framework}
In this section, we integrate SOA with the fast Fourier transform. It is noteworthy that although FFT is also employed in \cite{9}, their formulation is designed for single option pricing tasks, whereas we generalize the FFT-based algorithm to handle multiple options simultaneously. In detail, it is natural to consider the application of FFT for multiple-option pricing tasks, as these tasks involve evaluating a set of summations of the form in \eqref{eq:NumericalFormulaHatVk}. However, the framework is applicable only when all options share identical attributes, with strike prices being the sole varying attribute. Under this condition, pricing the entire option portfolio becomes equivalent to computing a discrete Fourier transform (DFT), which can be efficiently accelerated by FFT.

Suppose there are $L$ options with identical attributes, differing only in their strike prices. Denote the logarithms of these strike prices by $\tilde{k}_1,\tilde{k}_2,...,\tilde{k}_L$. Let $\underline{k}, \overline{k}\in\mathbb{R}$ be two numbers satisfying
\begin{equation*}
\underline{k} < k_{\min} < k_{\max} < \overline{k},
\end{equation*}
where
\begin{equation*}
k_{\min} = \min \{ \tilde{k}_1, \tilde{k}_2, \ldots, \tilde{k}_L \},\,\,k_{\max} = \max \{ \tilde{k}_1, \tilde{k}_2, \ldots, \tilde{k}_L \}.
\end{equation*}
In other words, an interval $\left[\underline{k}, \overline{k},\right]$ is needed to encompass the range $\left[k_{\min}, k_{\max}\right]$.
The lower bound $\underline{k}$ is specified by the model developer, whereas the upper bound $\overline{k}$ is determined adaptively according to the procedure described later.

For any log-strike price $k$, the approximation in \eqref{eq:NumericalFormulaHatVk} yields
\begin{equation}
\label{eq:EstimateVhat}
\hat{V}(k) \approx\frac{1}{\pi}\int_0^B e^{-izk}\eta(z)dz \approx \frac{\Delta z}{\pi}\sum_{j=0}^{N-1}w_je^{-ijk\Delta z}\eta(z_j),
\end{equation}
where $N$ cannot be chosen arbitrarily and must be determined according to specific considerations. For the moment, it is assumed that a suitable value of $N$ has been identified, with the grid spacing $\Delta z = \frac{B}{N-1}$.
The grid points for the log-strike prices are generated as
\begin{equation*}
k_n = \underline{k} + n\Delta k, n=0, 1, ..., N-1,
\end{equation*}
where the largest element, $k_{N-1}$, corresponds to the upper bound $\overline{k}$ introduced earlier.
Applying \eqref{eq:EstimateVhat} to $k_n$ yields
\begin{equation*}
\begin{aligned}
\hat{V}(k_n) \approx \frac{\Delta z}{\pi}\sum_{j=0}^{N-1}w_je^{-ij(\underline{k}+n\Delta k)\Delta z}\eta(z_j)
=\sum_{j=0}^{N-1}e^{-i\Delta k\Delta z nj}
\left[\frac{\Delta z}{\pi}w_je^{-ij\underline{k}\Delta z}\eta(z_j)\right].
\end{aligned}
\end{equation*}
To ensure that the above expression exactly matches the form of DFT, the grid spacings must satisfy
\begin{equation}
\label{eq:ProductOfDkDz}
\Delta k\Delta z = \frac{2\pi}{N}.
\end{equation}
We then define
\begin{equation*}
x_j = \frac{\Delta z}{\pi}w_je^{-ij\underline{k}\Delta z}\eta(z_j), j=0, 1,..., N-1.
\end{equation*}

Evaluating these summations can be accelerated by FFT, as $\left\{\hat{V}(k_n)\right\}_{n=0}^{N-1}$ represents the $N$-point DFT of $\left\{x_n\right\}_{n=0}^{N-1}$. Once $\left\{\hat{V}(k_n)\right\}_{n=0}^{N-1}$  is obtained, offset terms are added back to recover the option prices on the corresponding log-strike grids $\left\{k_n\right\}_{n=0}^{N-1}$.

Two remaining issues must be addressed. First, the log-strike prices of the option portfolio, denoted by $\left\{\tilde{k}_\ell\right\}_{\ell=0}^{L}$, do not coincide with the grid points $\left\{k_n\right\}_{n=0}^{N-1}$. To obtain prices at the desired strikes, linear interpolation is applied. Second, to ensure that the linear interpolation can be performed for all indices $\ell$, the upper bound of the strike grid must cover the maximum log-strike price. This requires that
\begin{equation}
\label{eq:ConditionOfDk}
\overline{k} = k_{N-1} > k_{\max} \Longrightarrow \underline{k} + (N-1)\Delta k > k_{\max}.
\end{equation}
By \eqref{eq:ProductOfDkDz}, we have
\begin{equation}
\label{eq:ConditionOfDk2}
\overline{k} = \underline{k} + \frac{(N-1)^2}{N}\frac{2\pi}{B} > k_{\max}.
\end{equation}

Inequality \eqref{eq:ConditionOfDk2} provides some insights into the efficiency of SOA when combined with FFT. Since SOA requires a smaller truncation point $B$ than CMA, its corresponding $\frac{2\pi}{B}$ is larger. Consequently, the minimum value of $N$ needed to satisfy \eqref{eq:ConditionOfDk2} is smaller for SOA, allowing its FFT computations to be completed faster.

Additionally, to mitigate potential numerical instabilities of the FFT near the boundaries, a symmetry constraint is imposed:
\begin{equation}
\label{eq:FFTSymmetricCondition}
k_{\min} - \underline{k} =  \overline{k} -  k_{\max}.
\end{equation}
Combining \eqref{eq:ConditionOfDk2} and \eqref{eq:FFTSymmetricCondition} yields
\begin{equation*}
\underline{k} = \frac{1}{2}\left[k_{\min}+k_{\max} -  \frac{(N-1)^2}{N}\frac{2\pi}{B}\right].
\end{equation*}
The optimal parameters identified in section \ref{subsec:SOANumExp} are $B^*=360$,  $N^*=576$ for CMA and $B^*=40$,  $N^*=64$ for SOA. The previously determined optimal parameters $B^*$ and $N^*$   remain applicable, as they satisfy all the above conditions.
\begin{figure}
	\centering
	\caption{Relationship between Prices Computed by CMA-FFT and CMA-OBO}
	\label{fig:FFTPriceCMAPriceCompare}
	\includegraphics[scale=0.6]{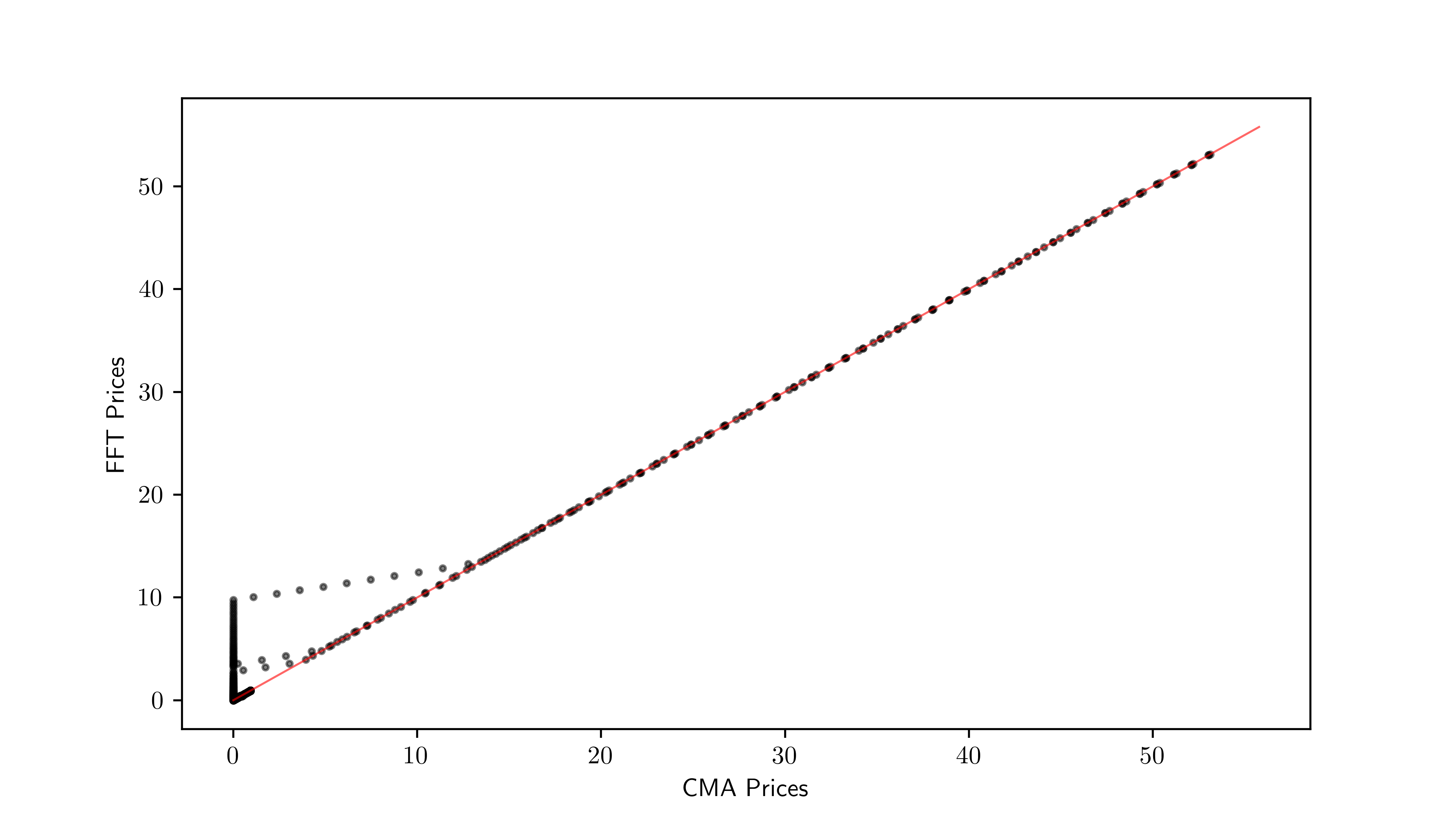}
\end{figure}
\begin{figure}
	\centering
	\caption{Relationship between Prices Computed by SOA-FFT and SOA-OBO}
	\label{fig:FFTPriceSOAPriceCompare}
	\includegraphics[scale=0.6]{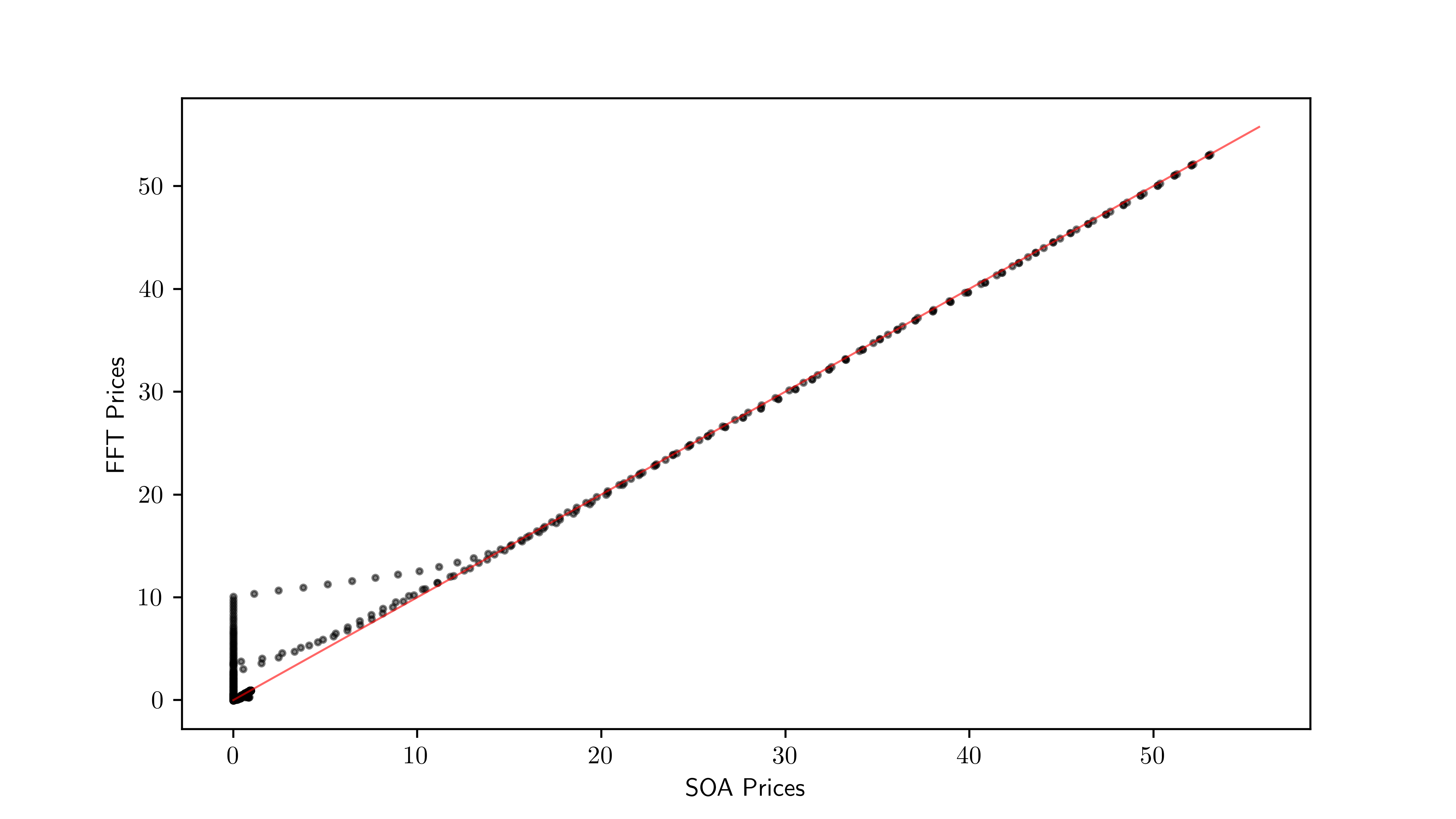}
\end{figure}
\subsection{Numerical Experiments for the FFT-based Framework}
The settings listed in Tables \ref{tab:DetailsOfOptions} and \ref{tab:DetailsOfStockPriceModels} are retained for the numerical experiments, except that the strike price is now varied from $50.0$ to $150.0$ with a step size of $1.0$. This configuration generates a portfolio consisting of $101$ options, which are priced using the FFT-based framework. 

To evaluate the accuracy of the FFT-based algorithms, two comparative analyses are conducted. The first comparison examines the pricing consistency between CMA-FFT and CMA-OBO (see Figure \ref{fig:FFTPriceCMAPriceCompare}), while the second compares SOA-FFT with SOA-OBO (see Figure \ref{fig:FFTPriceSOAPriceCompare}).

To evaluate computational efficiency, the execution times of CMA-FFT, SOA-OBO, and SOA-FFT are measured. Two pairwise comparisons are examined: $t_{\text{SOA-OBO}}$ \& $t_{\text{SOA-FFT}}$ and $t_{\text{CMA-FFT}}$ \& $t_{\text{SOA-FFT}}$.
 The results are presented in Figures \ref{fig:SOAOBOTimeCost} and \ref{fig:SOACMATimeCost}, respectively.

Two linear regression models (without intercept) are fitted to estimate the ratios 
\begin{equation}
\label{eq:FFTOLSEqExeTime}
\begin{aligned}
t_{\text{SOA-FFT}} = \beta_B t_{\text{SOA-OBO}} + \varepsilon,\\
t_{\text{SOA-FFT}} = \beta_C t_{\text{CMA-FFT}} + \varepsilon.
\end{aligned}
\end{equation}

The FFT-based framework shows numerical instability when pricing out-of-the-money options. This is evidenced by Figure \ref{fig:FFTPriceCMAPriceCompare} and \ref{fig:FFTPriceSOAPriceCompare}, where data points corresponding to higher option values lie close to the 45-degree line, indicating strong agreement between FFT prices and those obtained via standard CMA or SOA. In contrast, data points with lower option values deviate substantially from the 45-degree line, revealing pronounced pricing errors in certain out-of-the-money cases.
Furthermore, the SOA integrated with FFT not only shortens the execution time by over 97\% compared with the standard SOA, greatly improving computational efficiency, but also consistently outperforms CMA when combined with FFT, with the execution time of SOA-FFT being approximately 57.58\% of that of CMA-FFT.
\begin{figure}[h]
	\centering
	\caption{Relationship between $t_{\text{SOA-OBO}}$ and $t_{\text{SOA-FFT}}$}
	\label{fig:SOAOBOTimeCost}
	\includegraphics[scale=0.6]{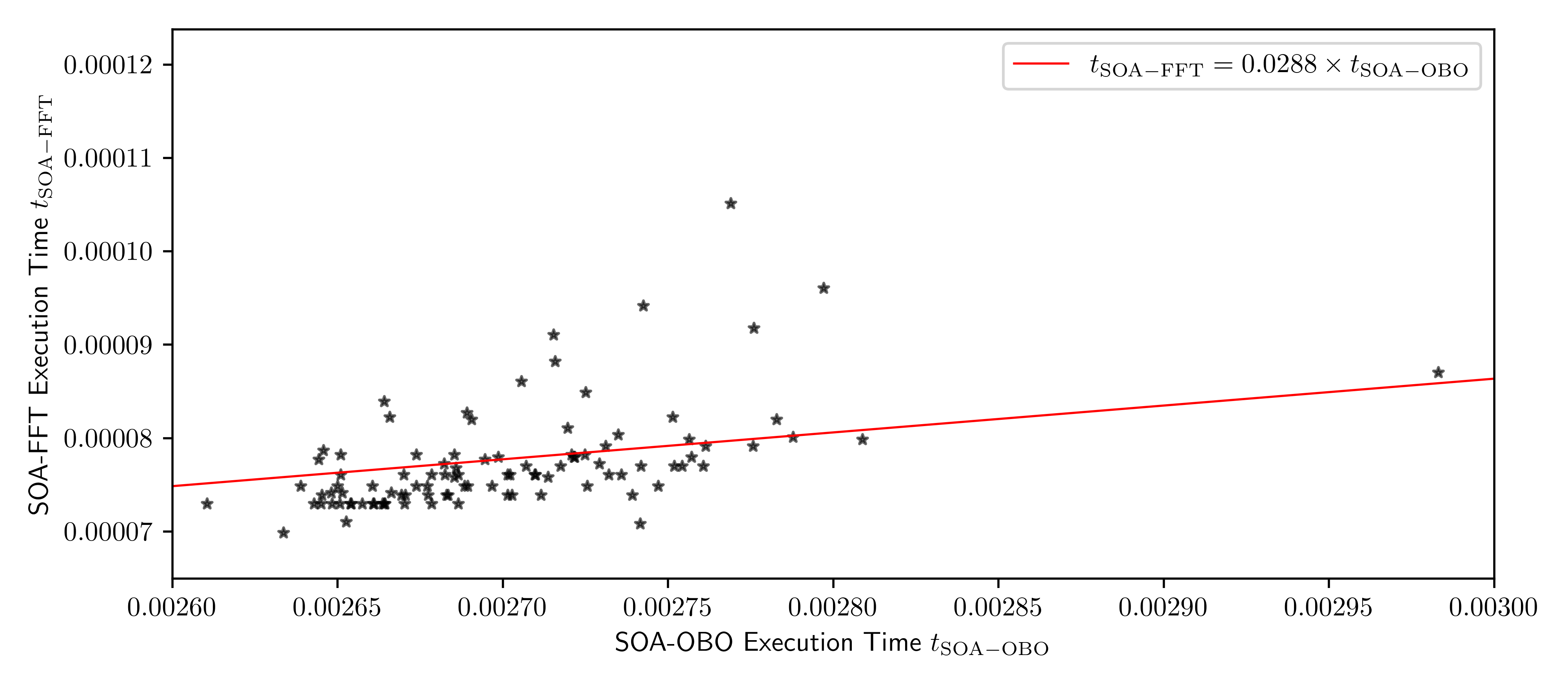}
\end{figure}
\begin{figure}[h]
	\centering
	\caption{Relationship between $t_{\text{CMA-FFT}}$ and $t_{\text{SOA-FFT}}$}
	\label{fig:SOACMATimeCost}
	\includegraphics[scale=0.6]{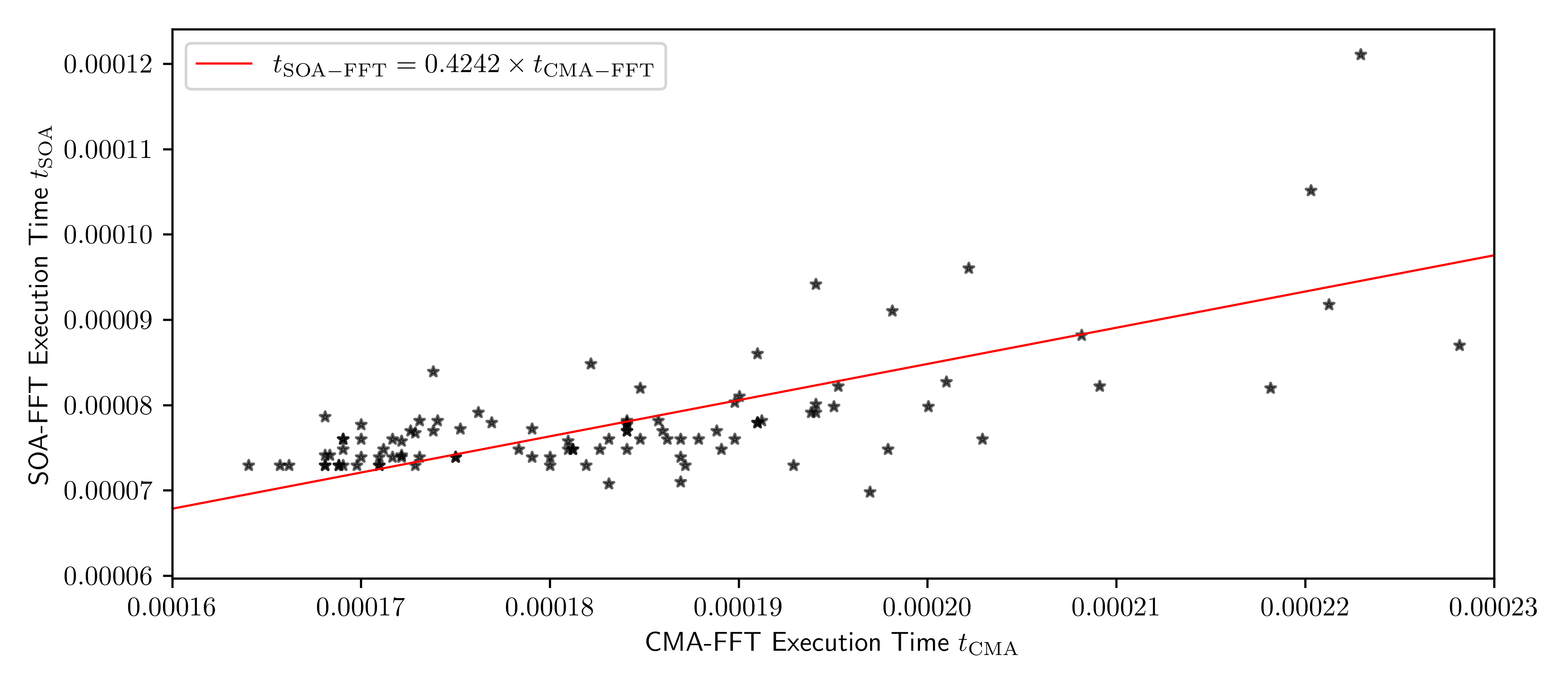}
\end{figure}

\section{ML for Multiple Options Pricing
	\label{sec:MLFrameworksForPricingMultOptions}}
\subsection{ML-based Framework}
ML techniques have been extensively applied to option pricing, where a trained ML model maps a feature vector $\bm{x}$ representing the set of attributes influencing the option price to a scalar output $p(\bm{x})$, the option price. The proposed ML-based framework comprises two components: the training stage and the forecasting stage, as illustrated in Figure \ref{fig:MLFrameworkTraining}.
\begin{figure}[h]
	\centering
	\caption{The ML-based Framework}
	\label{fig:MLFrameworkTraining}
	\begin{minipage}{0.5\textwidth}
		\centering
		\sffamily Training Stage
	\end{minipage}\\\quad\\
	\begin{tikzpicture}[
	node distance=0.4cm,
	every node/.style={font=\rmfamily\normalsize, align=center},
	process/.style={draw, rectangle, rounded corners=2pt, minimum width=1.1cm, minimum height=1cm, thick},
	arrow/.style={thick, -{Stealth[length=3mm,width=2mm]}}
	]
	
	\node[process] (gen) {Random Option\\Contracts Generator};
	\node[ right=of gen] (x) {$\{\bm{x}_n\}_{n=1}^N$};
	\node[process, right=of x] (soa) {SOA\\Pricer};
	\node[right=of soa] (y) {$\{y_n\}_{n=1}^N$};
	\node[process, right=of y] (train) {Training\\Dataset\\$\{\bm{x}_n, y_n\}_{n=1}^N$};
	\node[process, right=of train] (ml) {ML\\Algorithm};
	\node[right=of ml] (m) {$\mathcal{M}_\theta$};
	
	\draw[arrow] (gen)--(x);
	\draw[arrow] (x)--(soa);
	\draw[arrow] (soa)--(y);
	\draw[arrow] (y)--(train);
	\draw[arrow] (train)--(ml);
	\draw[arrow] (ml) -- (m);
	
	\draw[arrow] ([yshift=-0.22cm]gen.south) |- ++(2,-0.3) -| (train.south);
	\end{tikzpicture}\\\quad\\\quad\\
	\begin{minipage}{0.6\textwidth}
		\centering
		\sffamily Forecasting Stage
	\end{minipage}\\\quad\\
	\begin{tikzpicture}[
	node distance=0.5cm,
	every node/.style={font=\rmfamily\normalsize, align=center},
	process/.style={draw, rectangle, rounded corners=2pt, minimum width=1.2cm, minimum height=1cm, thick},
	arrow/.style={thick, -{Stealth[length=3mm,width=2mm]}}
	]
	
	\node[process] (option) {New Option\\Contracts};
	\node[ right=of option] (m) {$\mathcal{M}_{\theta}$};
	\node[process, right=of m] (price) {Theoretical\\Prices};
	
	\draw[arrow] (option)--(m);
	\draw[arrow] (m)--(price);
	\end{tikzpicture}
\end{figure}
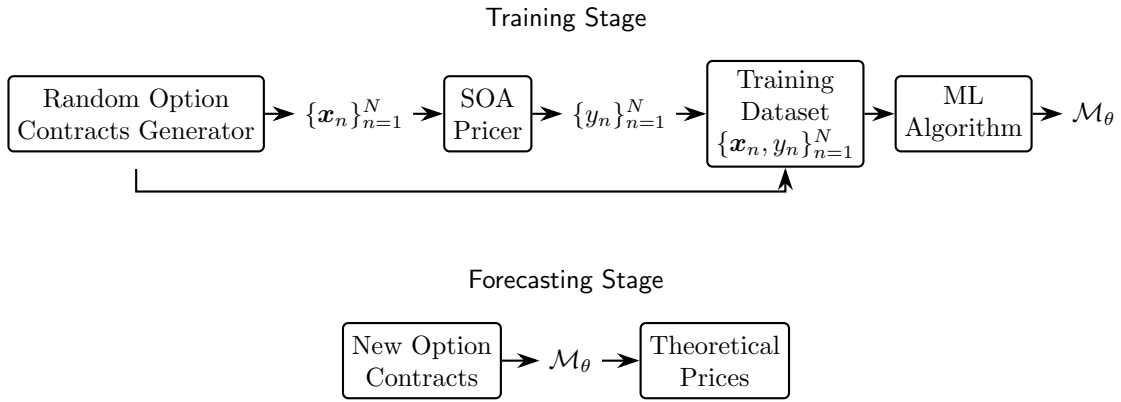

The objective of the training stage is to construct the training dataset and train a ML model capable of accurately mapping any feature vector $\boldsymbol{x}$ to its corresponding option price $y$.
A sufficient number of data points must be collected prior to model training. The SOA algorithm is used to generate synthetic option contracts.
The feature vectors incorporate all relevant attributes of the option contracts: the option type (European or digital), spot price $S_0$, strike price $K$, maturity $T$, risk-free rate $r$, and the parameters of the underlying stock price models $(\sigma, \kappa, \theta, \rho, V_0, \nu)$.

However, some attributes operate on vastly different numerical scales. Attributes $T, r$ and the parameters of stock price models are typically smaller than $1$, whereas $S_0$, $K$ can reach much larger magnitudes. For example, the close price of NVIDIA stock on August 29, 2025 was $\$ 174.18$, and strike prices of call options on NVIDIA shares ranged from \$$50$ to \$$365$. Such discrepancies in scale must be mitigated to ensure that all attributes contribute comparably and to accelerate the convergence of optimization algorithms such as gradient descent.

To achieve this normalization, we fix the spot price at $S_0 = 1$ and transform the strike price as $K^{\prime} = \frac{K}{S_0}$. This procedure is equivalent to rescaling both $S_0$ and $K$ by the factor $\frac{1}{S_0}$, ensuring that all attributes lie on comparable numerical scales. Correspondingly, the outcomes must also be rescaled: for European options, the rescaled prices are the actual prices multiplied by $\frac{1}{S_0}$, whereas for digital options, rescaling leaves the prices unchanged.
Furthermore, the option type is encoded as a binary variable:
\begin{equation*}
\text{OpType} = \left\{
\begin{aligned}
&1, \text{\,\,if the option is European}, \\
&0, \text{\,\,if the option is digital}.
\end{aligned}
\right.
\end{equation*}
After rescaling $S_0$, $K$ and encoding the option type, the resulting feature vector (denoted by $\bm{x}\in\mathbb{R}^{11}$) is
\begin{equation*}
\bm{x} = \left[\text{OpType}, K^\prime, T, r, \sigma, \kappa, \theta, \rho, V_0, \nu\right].
\end{equation*}

The outcome corresponding to each feature vector $\bm{x}$ (denoted by $y$) represents the rescaled option price computed by SOA, which is adopted for its superior efficiency over CMA. The actual option price $y^{\text{actual}}$ can be recovered from $y$ through
\begin{equation*}
\label{eq:MLTechniquesActualPriceRecover}
y^{\text{actual}} = \left[\left(S_0-1\right)\text{OpType} + 1\right] y.
\end{equation*}
The feature vector $\bm{x}$ includes parameters from all stock price models, resulting in redundant elements for a given data point. For example, when the stock prices follow EVGP, only $\sigma$, $\theta$, and $\nu$ are relevant, whereas $\kappa$, $\rho$, and $V_0$ are redundant and thus are set to $-1$. Accordingly, for a European option under EVGP, the feature vector takes the form
\begin{equation}
\bm{x} = \left[1, K^\prime, T, r, \sigma, -1, \theta, -1, -1, \nu\right].
\end{equation}
Other combinations of option types and stock price models are handled analogously.

Data points are generated by the random option contract generator in Figure \ref{fig:MLFrameworkTraining}. All option attributes are generated by random sampling. Specifically, OpType is drawn from a Bernoulli distribution, $\text{Bernoulli}\left(\tfrac{1}{2}\right)$, ensuring an equal number of European and digital options in the dataset. For each remaining attribute, lower and upper bounds $\underline{L}$ and $\overline{L}$ are specified, and samples are drawn from a uniform distribution over the interval $\left[\underline{L}, \overline{L}\right]$. The bounds used in the numerical experiments are summarized in Table \ref{tab:IntervalsOfInputVariablesInTheRandomSampling}.

To ensure that the trained ML models achieve sufficiently small pricing errors across the entire input domain defined in Table \ref{tab:IntervalsOfInputVariablesInTheRandomSampling}, a large number of data points are required. Accordingly, we generate $N^{\text{train}} = 2\times 10^6$ feature vectors $\left\{\bm{x}_n\right\}_{n=1}^{N^{\text{train}}}$, each associated with a target value $y_n$ computed via the SOA pricer in Figure \ref{fig:MLFrameworkTraining}. This yields the training dataset $\left\{\bm{x}_n, y_n\right\}_{n=1}^{N^{\text{train}}}$.

The test dataset is constructed in the same manner and contains $N^{\text{test}}=10000$ options, and the forecasting stage aims to compute the prices of these options using the trained ML model $\mathcal{M}_\theta$. In practical applications, however, the test dataset corresponds to real option contracts observed in production environments rather than simulated data.
\begin{table}
	\renewcommand{\arraystretch}{1.3} 
	\centering
	\caption{Bounds of Input Variables in the Random Sampling}
	\label{tab:IntervalsOfInputVariablesInTheRandomSampling}
	\begin{tabular}{cccccc}
		\hline
		Input Variable & Lower Bound $\underline{L}$ & Upper Bound $\overline{L}$ \\
		\hline
		$S_0$ & $140.0$ & $160.0$ \\
		$K$  & $0.95S_0$ & $1.05S_0$ \\
		$T$  &  $0$ months $(0.00)$ & $12$ months $(1.00)$ \\
		$r$  & $0.0\%$ & $5.0\%$ \\
		$\sigma$& $0.00$ & $0.20$ \\
		$\kappa$& $0.00$ & $0.20$ \\ 
		$\theta$& $0.00$ & $0.20$ \\
		$\rho$& $0.00$ & $0.20$ \\
		$V_0$& $0.00$ & $0.20$ \\
		$\nu$& $0.00$ & $0.20$ \\
		\hline
	\end{tabular}
\end{table}

\subsection{Numerical Experiments for the ML-based Framework}
\subsubsection{Training Neural Networks}
An NN takes a feature vector $\bm{x}$ as input and transforms it into an output through a sequence of weighted linear operations and nonlinear activation functions. Various activation functions have been proposed in the literature, each introducing different forms of nonlinearity,
and the Leaky ReLU function is adopted in this paper. For vector inputs, $S$ is applied elementwise. The architecture of the NN employed in the numerical experiments comprises 11 layers in total, which includes one input dense layer with an input dimension of 11, five hidden dense layers, five Leaky ReLU activation layers, and one output dense layer with an output dimension of 1. Each of the hidden dense layers is configured with 128 neurons.

To train the NN efficiently, the mini-batch gradient descent algorithm is employed with a batch size of $256$.
The training process spans $3000$ epochs to ensure sufficient convergence.
To ensure numerical stability and attain high predictive accuracy, a carefully tuned learning rate schedule is employed. The initial learning rate is set to $3\times10^{-6}$ and is decayed by a factor of $0.1$ after every $500$ training epochs.
This configuration enables gradual and precise updates of the NN parameters during optimization.

Figure \ref{fig:MLNNMeanMSE} illustrates the evolution of the mean MSE loss throughout the training process. For each epoch, the mean MSE loss is calculated as the average of the MSE losses across all mini-batches. As expected, the loss exhibits a consistent downward trend, indicating effective learning and stable convergence.
Once the NN is fully trained, it is applied to forecast out-of-sample option prices for data points in the test dataset.
These predicted prices are then compared with the true values obtained from the SOA. As shown in Figure \ref{fig:MLNNPriceCompare.png}, nearly all points lie close to the 45-degree line, 
demonstrating that the predictions align closely with the actual outcomes and confirming the high accuracy of the NN.
\begin{figure}
	\centering
	\caption{Mean MSE Losses of All Epochs (Log Scale)}
	\label{fig:MLNNMeanMSE}
	\includegraphics[scale=0.6]{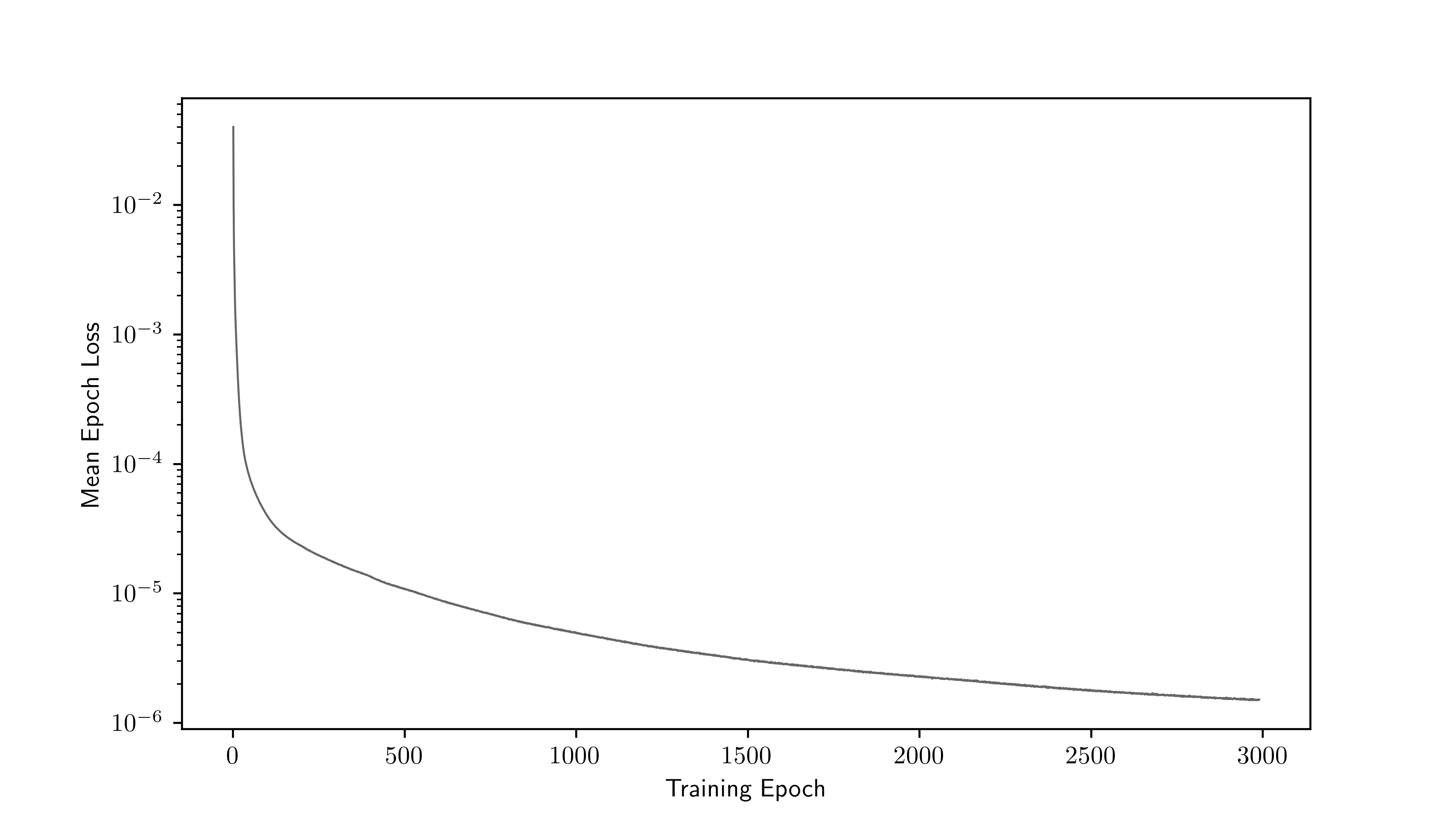}
\end{figure}
\begin{figure}[H]
	\centering
	\caption{Out-of-sample Relationship between NN Predicted Prices and SOA Actual Prices}
	\label{fig:MLNNPriceCompare.png}
	\includegraphics[scale=0.6]{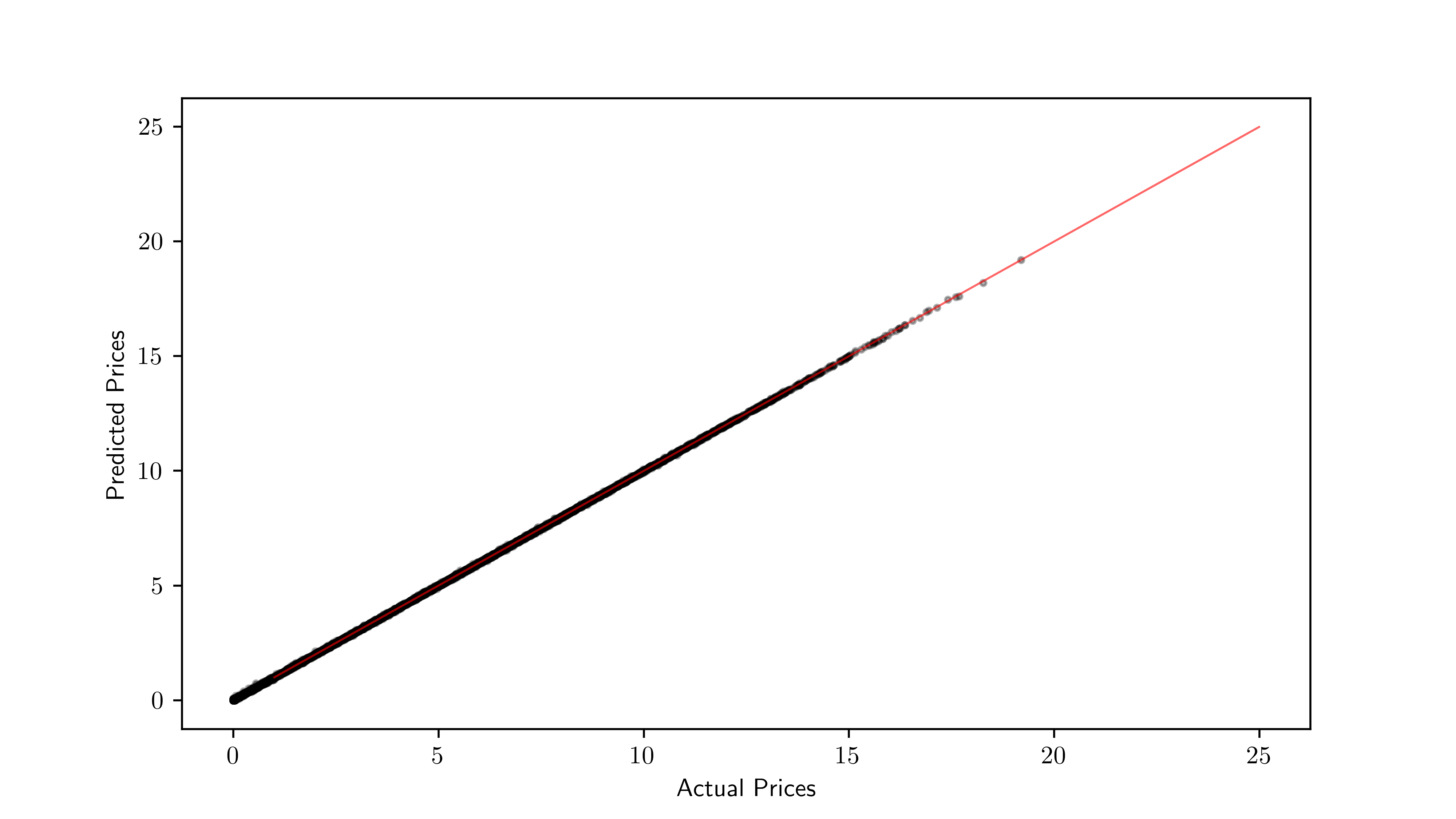}
\end{figure}

To measure the accuracy, let $\hat{y}_n$ and $y^{\text{actual}}_n$ be the $n^{\text{th}}$ predicted and actual prices of the $n^{\text{th}}$ option in the test dataset, respectively, where $n=1,2,...,N^{\text{test}}$. Then for any ML model, the absolute pricing error is defined as
\begin{equation*}
\text{absolute pricing error} = \sqrt{
	\frac{1}{N^{\text{test}}} \sum_{n=1}^{N^{\text{test}}}\left(
	\hat{y}_n - y^{\text{actual}}_n
	\right)^2},
\end{equation*} 
and the relative pricing error is defined as
\begin{equation*}
\text{relative pricing error} = 
\frac{1}{N^{\text{test}}} \sum_{n=1}^{N^{\text{test}}}\left\lvert
\frac{\hat{y}_n}{y^{\text{actual}}_n} - 1
\right\rvert.
\end{equation*} 
For the fully trained NN, the absolute and relative pricing errors are $0.0009$ and $0.0417\%$, respectively.

\subsubsection{Training Ensemble Learning Models: Bagging and Boosting}
Ensemble learning algorithms that utilize decision trees (DTs) as weak learners can generally be classified into two main categories: bagging and boosting. In this study, both approaches are examined. RFs represent the bagging paradigm, and GBDTs represent the boosting paradigm.

Overfitting and underfitting in RFs and GBDTs can be effectively mitigated by controlling the maximum depth of individual decision trees, which serves as a key hyperparameter. Shallow trees may fail to capture the underlying data patterns adequately, leading to underfitting. But excessively deep trees risk overfitting and substantially increasing computational costs.
To balance the model complexity and the generalization capability, the optimal value of the maximum depth is determined through grid search combined with cross-validation. The results indicate that the optimal maximum depth is $20$ for RF and $15$ for GBDT.
Figure \ref{fig:MLRFCV} and \ref{fig:MLGBDTCV} illustrate the relationship between the mean MSE loss obtained from cross-validation and the corresponding maximum depth.
As shown in these figures, increasing the maximum depth beyond $20$ for RF or $15$ for GBDT yields no substantial improvement in the mean MSE loss. In some cases, it even leads to worse mean MSE values due to overfitting.

\begin{figure}
	\centering
	\caption{Relationship between Cross-Validation MSE Loss (Log Scale) and RF Max Depth}
	\label{fig:MLRFCV}
	\includegraphics[scale=0.6]{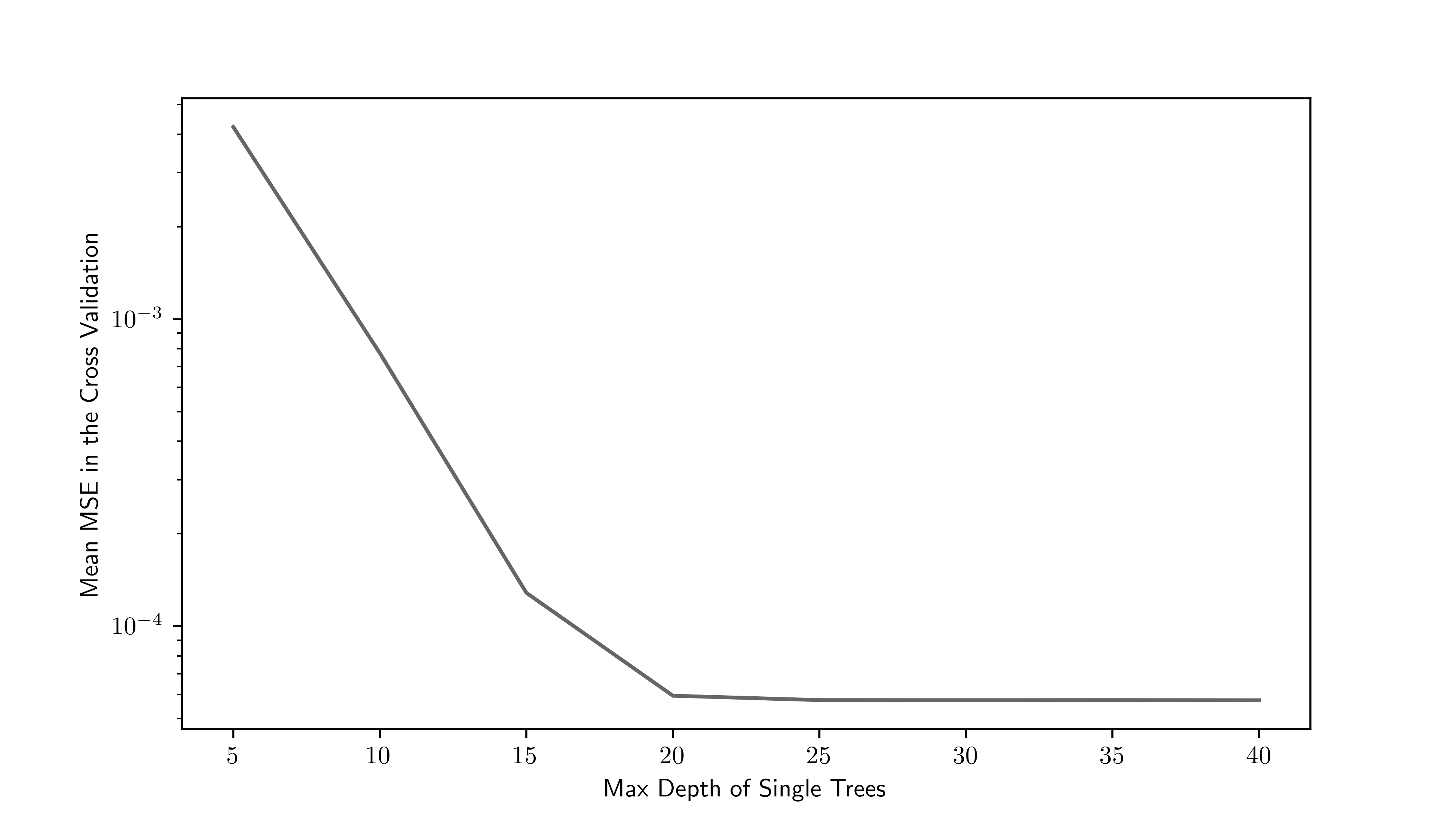}
\end{figure}
\begin{figure}
	\centering
	\caption{Relationship between Cross-Validation MSE Loss (Log Scale) and GBDT Max Depth}
	\label{fig:MLGBDTCV}
	\includegraphics[scale=0.6]{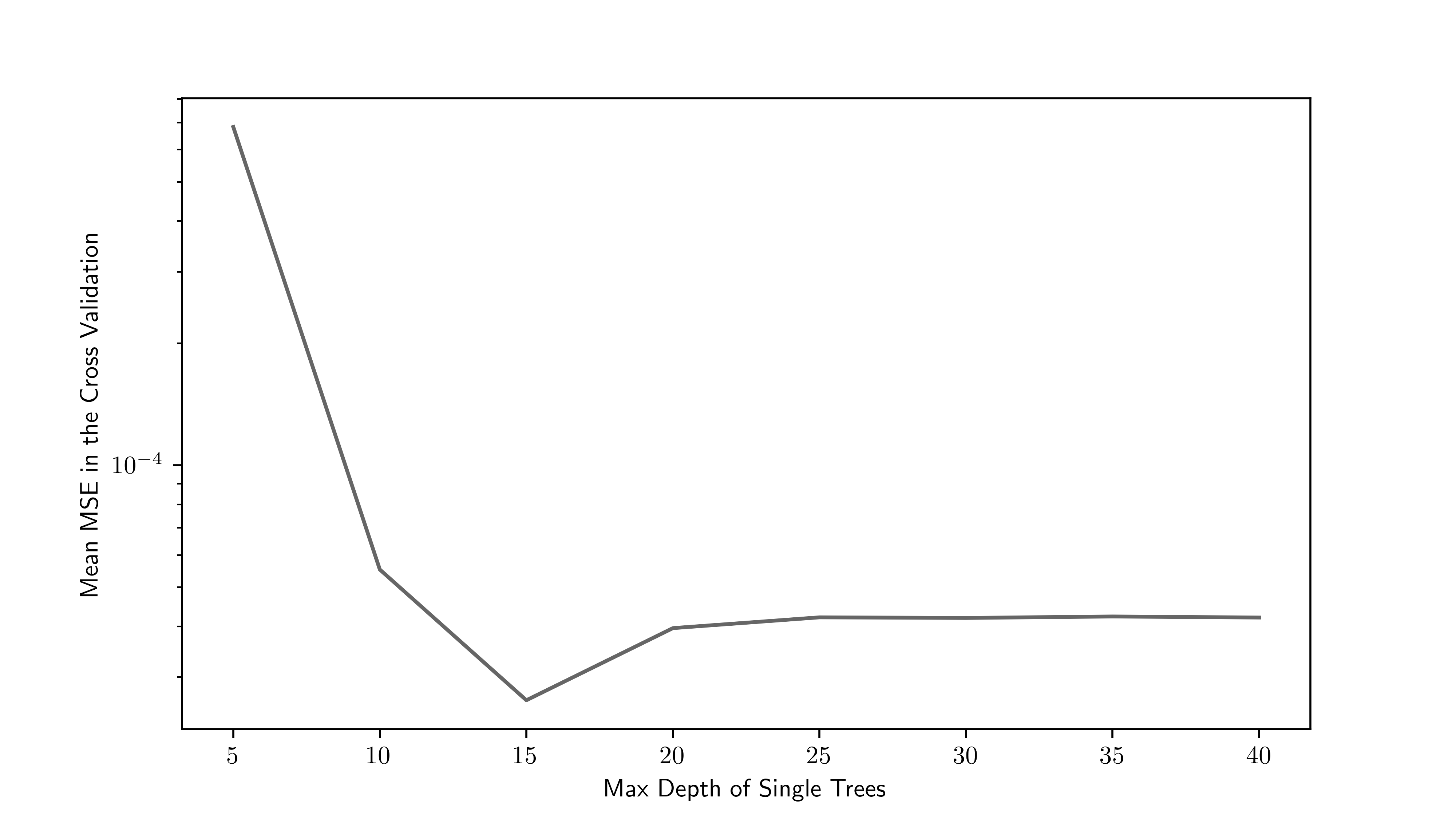}
\end{figure}

During the training of RFs and GBDTs, the number of weak learners is set to $100$. The bootstrapping hyperparameter is fixed at $0.7$, indicating that each weak learner is trained on a randomly selected subset comprising 70\% of the total training dataset ($0.7N^{\text{train}}$). Figure \ref{fig:MLRFPriceCompare.png} compares the prices predicted by RF with those computed by SOA, while Figure \ref{fig:MLGBDTPriceCompare.png} presents analogous results for GBDT. Both RF and GBDT exhibit high predictive accuracy for out-of-sample data. Consistent with the NN results, nearly all data points in Figure \ref{fig:MLRFPriceCompare.png} and \ref{fig:MLGBDTPriceCompare.png} lie close to the 45-degree line, confirming strong agreement between predicted and actual prices.

The absolute pricing errors are $0.0044$ for RF and $0.0008$ for GBDT, while the corresponding relative pricing errors are $0.1884\%$ and $0.0668\%$, respectively. 
\begin{figure}
	\centering
	\caption{Out-of-sample Relationship between RF Predicted Prices and SOA Actual Prices}
	\label{fig:MLRFPriceCompare.png}
	\includegraphics[scale=0.6]{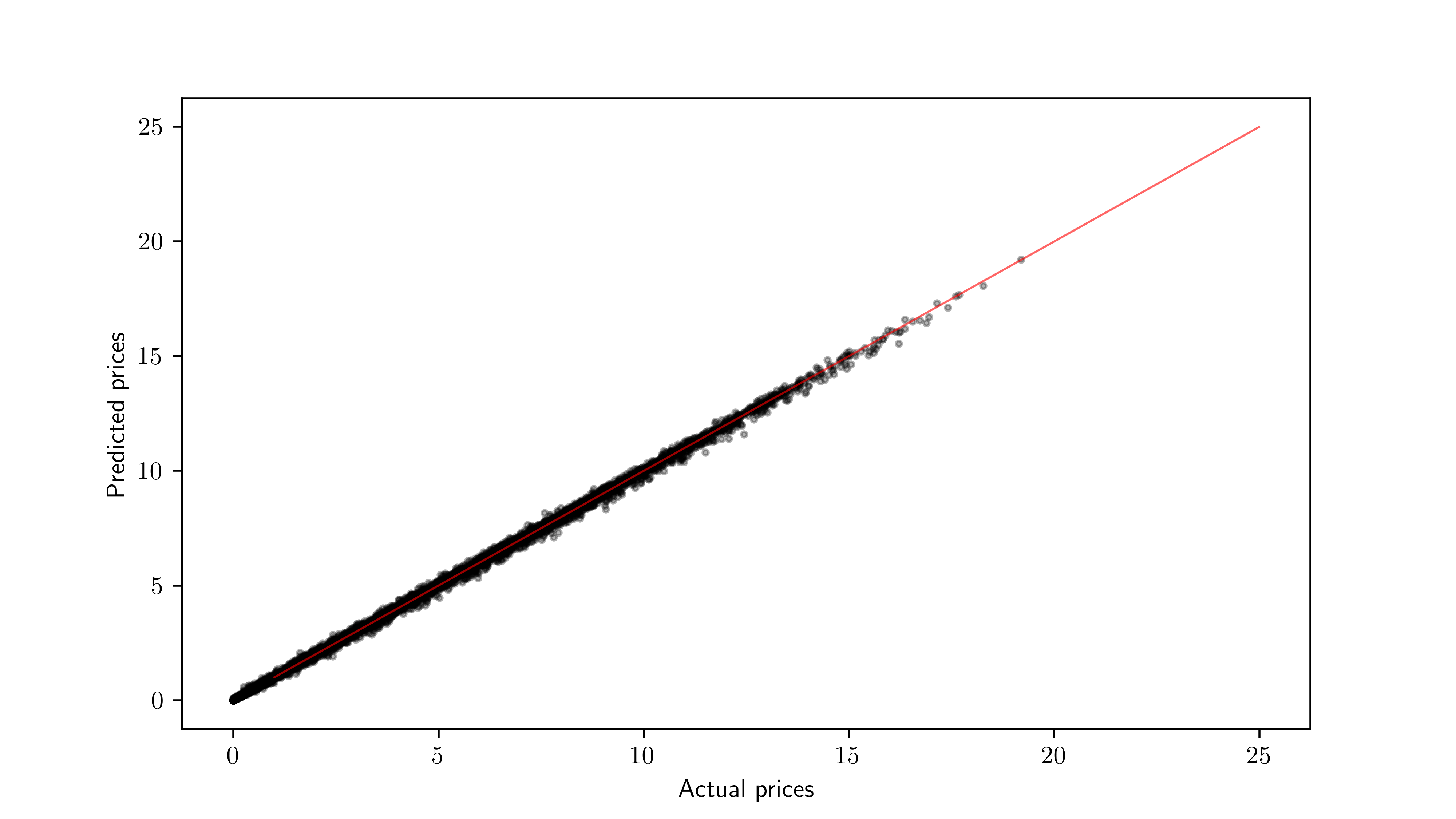}
\end{figure}
\begin{figure}
	\centering
	\caption{Out-of-sample Relationship between GBDT Predicted Prices and SOA Actual Prices}
	\label{fig:MLGBDTPriceCompare.png}
	\includegraphics[scale=0.6]{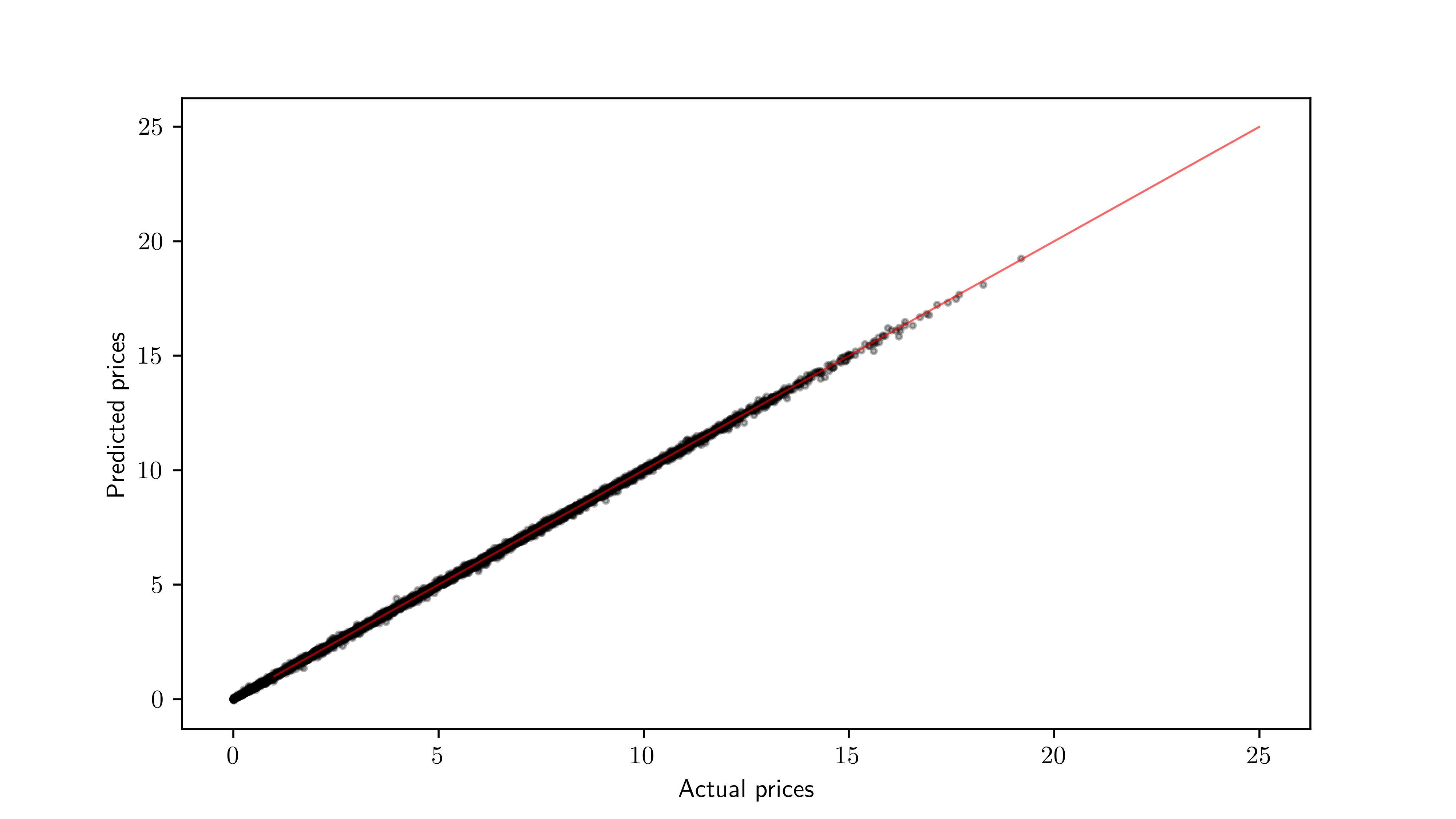}
\end{figure}
\subsubsection{Execution Times}
The execution times of the three ML algorithms are compared against that of SOA. The execution time of each algorithm is defined as the total time required to price all $10000$ options in the test dataset. To mitigate the influence of hardware-induced randomness, the pricing process for the entire test dataset is repeated $100$ times for each algorithm.
Figure \ref{fig:MLNNExeTimeSOAExeTime}, \ref{fig:MLRFExeTimeSOAExeTime} and \ref{fig:MLGBDTExeTimeSOAExeTime} present the comparative execution times of NN, RF, and GBDT relative to standard SOA, respectively.

\begin{figure}
	\centering
	\caption{Relationship between NN Execution Time and SOA Execution Time}
	\label{fig:MLNNExeTimeSOAExeTime}
	\includegraphics[scale=0.6]{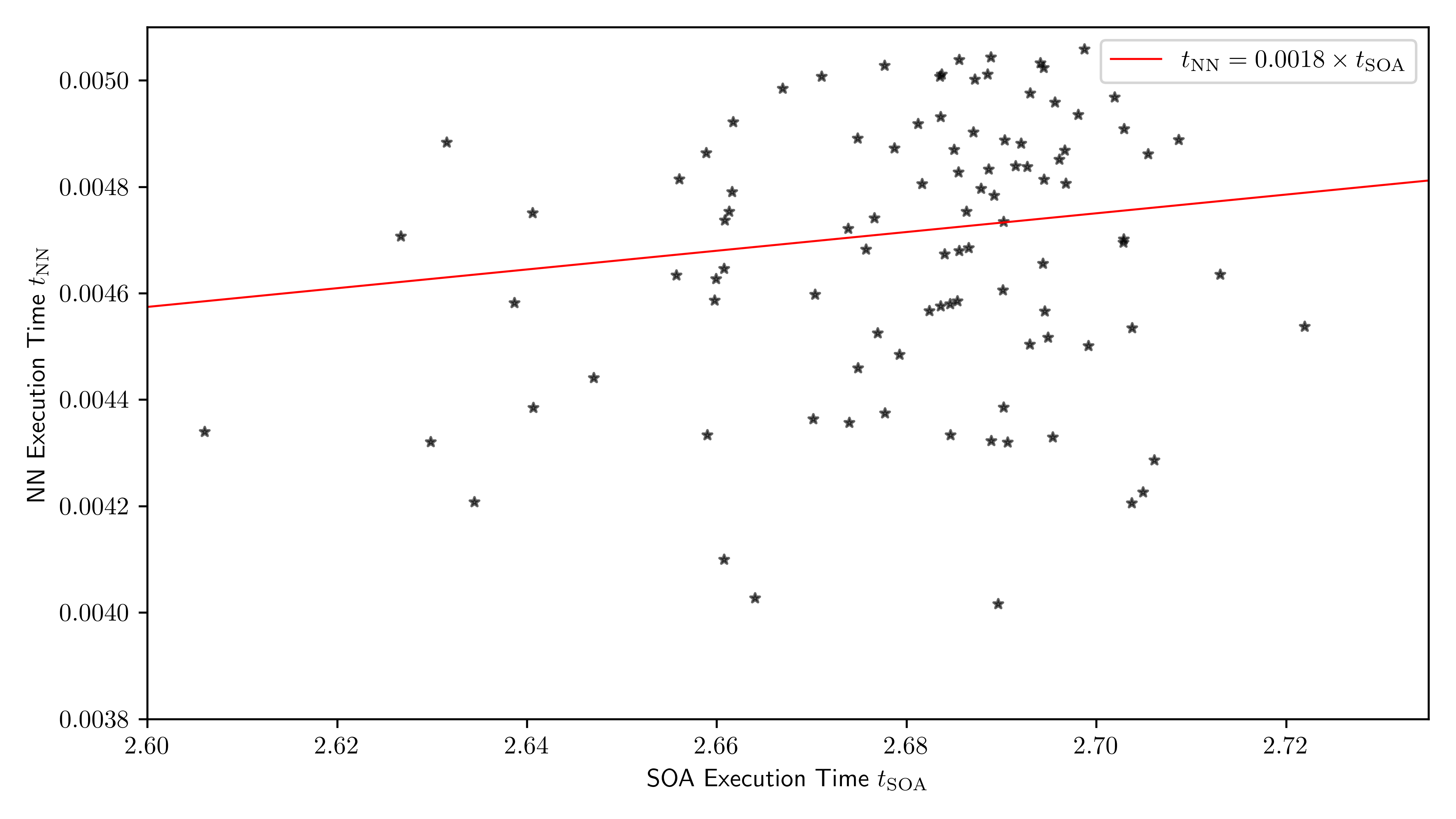}
\end{figure}
\begin{figure}
	\centering
	\caption{Relationship between RF Execution Time and SOA Execution Time}
	\label{fig:MLRFExeTimeSOAExeTime}
	\includegraphics[scale=0.6]{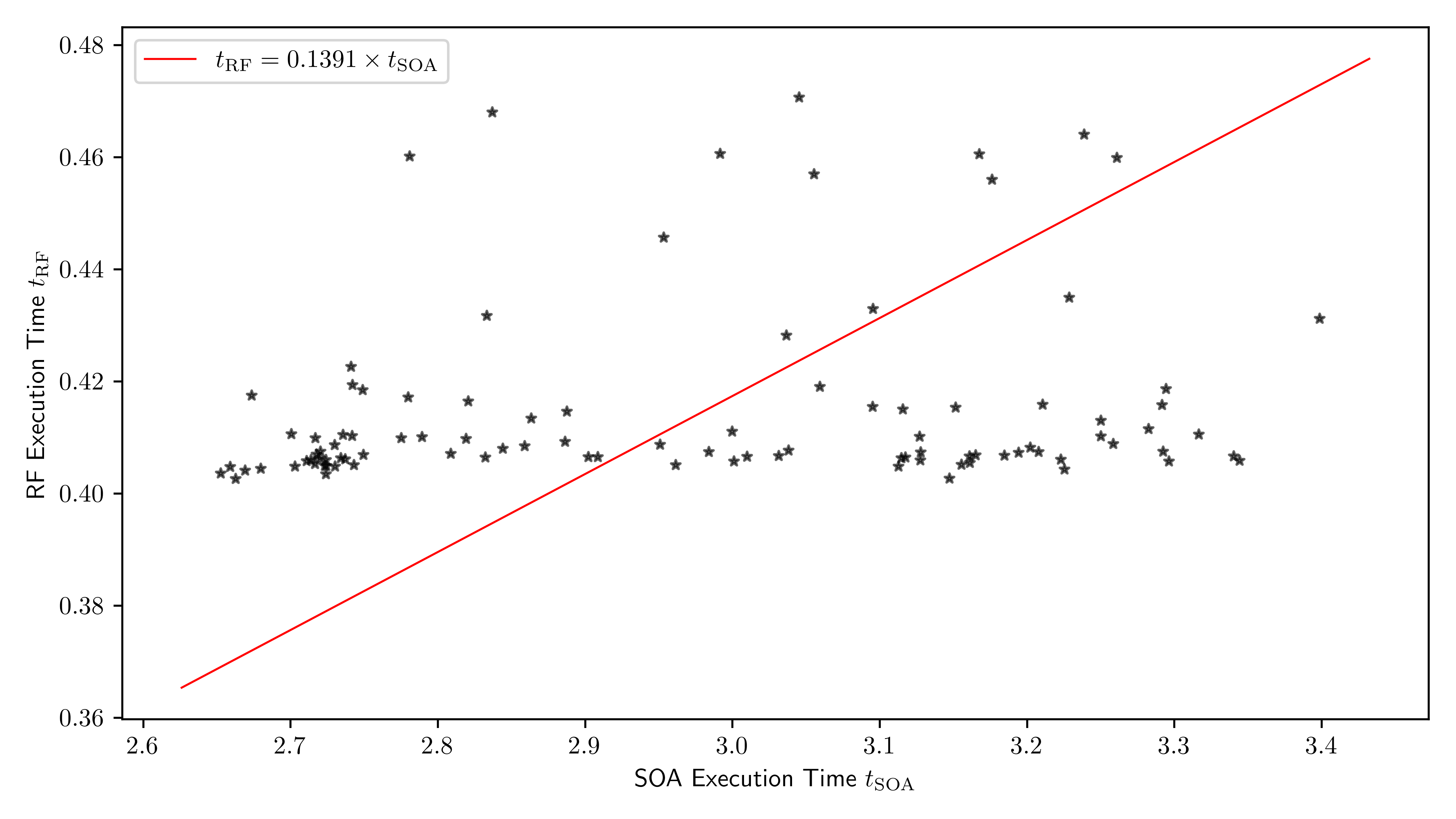}
\end{figure}
\begin{figure}
	\centering
	\caption{Relationship between GBDT Execution Time and SOA Execution Time}
	\label{fig:MLGBDTExeTimeSOAExeTime}
	\includegraphics[scale=0.6]{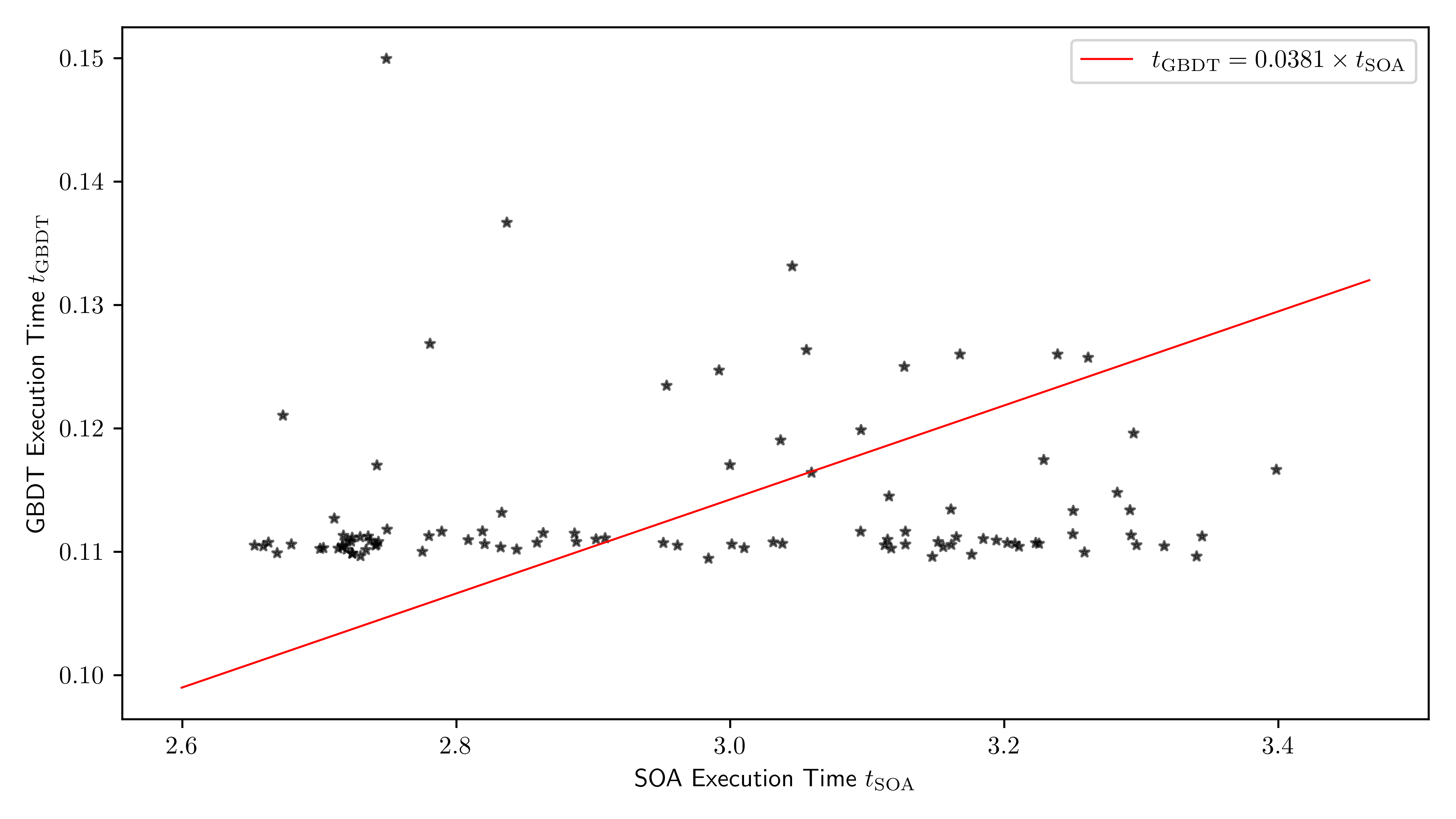}
\end{figure}

Let $t_{\text{NN}}, t_{\text{RF}}, t_{\text{GBDT}}$ and $t_{\text{SOA}}$ be the execution times of NN, RF, GBDT, and SOA, respectively.
To quantify their relative computational efficiency, three linear regression models (without intercept) are fitted to estimate $\frac{t_{\text{NN}}}{t_{\text{SOA}}}$, $\frac{t_{\text{RF}}}{t_{\text{SOA}}}$ and $\frac{t_{\text{GBDT}}}{t_{\text{SOA}}}$:
\begin{equation}
\label{eq:MLOLSEqExeTime}
\begin{aligned}
&	t_{\text{NN}} = \beta_A t_{\text{SOA}} + \varepsilon, \\
&	t_{\text{RF}} = \beta_B t_{\text{SOA}} + \varepsilon, \\
&	t_{\text{GBDT}} = \beta_C t_{\text{SOA}} + \varepsilon. \\
\end{aligned}
\end{equation}
Regression results are summarized in Table \ref{tab:MLOLSEqExeTime}. The findings indicate that the execution times of NN, RF, and GBDT are approximately $0.18\%, 13.91\%$ and $3.81\%$ of SOA, respectively.
Accordingly, the improvements in computational efficiency are $99.82\%$ for NN, $86.09\%$ for RF and $96.19\%$ for GBDT.
All estimated regression coefficients are statistically significant and exhibit narrow confidence intervals, confirming the robustness and reliability of these results.

\begin{table}
	\renewcommand{\arraystretch}{1.3}
	\centering
	\begin{threeparttable}
		\caption{Regression Results for \eqref{eq:MLOLSEqExeTime} }
		\label{tab:MLOLSEqExeTime}
		\begin{tabular}{ccccccccccccccccc}
			\hline
			Coefficient & Estimate & Standard Error & $t-$Statistic & $p-$Value & $95\%$ Confidence Interval &  \\
			\hline
			$\beta_A$ & $0.0018$ & $1.57\times 10^{-5}$&$112.034$ & $4.3076\times 10^{-106}$& $[0.0017, 0.0018]$\\
			$\beta_B$ & $0.1391$ & $1.07\times 10^{-3}$&$129.181$ & $3.5673\times 10^{-112}$& $[0.1370, 0.1413]$\\
			$\beta_C$ & $0.0381$ & $3.44\times 10^{-4}$&$110.521$ & $1.6368\times 10^{-105}$& $[0.0374, 0.0387]$\\
			\hline
		\end{tabular}
	\end{threeparttable}
\end{table}

Table \ref{tab:PerformancesOfDifferentMLAlgorithms} summarizes the performance of the three ML algorithms, including their absolute pricing errors, relative pricing errors, and execution times. 
Among the evaluated algorithms, NN demonstrates the lowest absolute and relative pricing errors. It achieves computational efficiency comparable to the FFT-based framework which achieves an improvement in execution time of approximately  $97.12\%$.
While NNs demonstrate outstanding predictive performance and fast execution speed, their training process is considerably more computationally demanding than RF or GBDT. In comparison, GBDTs achieve similar gains in execution efficiency with much lower training complexity. Furthermore, GBDTs offer superior model interpretability and enhanced control over overfitting relative to NNs.
Therefore, both NN and GBDT are recommended for real-world applications, depending on the specific requirements and hardware availability. 
When a sufficient number of data points are available and Graphics Processing Unit (GPU) resources are abundant, NN is undoubtedly the preferred choice. Otherwise, it is advisable for model developers to place greater emphasis on GBDT.
\begin{table}[t]
	\renewcommand{\arraystretch}{1.3} 
	\centering
	\caption{Performances of Three ML Algorithms}
	\label{tab:PerformancesOfDifferentMLAlgorithms}
	\begin{tabular}{cccccc}
		\hline
		ML Algorithm & Absolute Pricing Error & Relative Pricing Error & Improvement in Execution Time \\
		\hline
		NN  &$0.0005$ & $4.17$ bps & 99.82\%\\
		RF  &$0.0044$ & $18.84$ bps & 86.09\%\\
		GBDT&$0.0008$ & $6.68$ bps  & 96.19\%\\
		\hline
	\end{tabular}
\end{table}

While the FFT-based framework also delivers high computational speed, the ML-based framework demonstrates superior numerical stability and pricing accuracy. 
Furthermore, the ML-based framework is less restrictive regarding input option attributes than the FFT-based framework.
Overall, the results indicate that the ML-based framework provides the most balanced performance between accuracy and efficiency, making it the most suitable approach for large-scale option pricing applications based on the FT.

\section{Conclusions\label{sec:Conclusions}}
This paper presents an efficient ML-based framework for multiple option pricing tasks where the options are path-independent and the underlying stock prices follow exponential L{\'e}vy processes.
The data generation is performed using an improved FT-based algorithm, termed SOA, developed as an enhancement of CMA. SOA leverages the theoretical relationship between the smoothness of a function and the tail decay rate of its FT,  introducing a smooth offset term that replaces the original offset in CMA. 
This modification substantially reduces computational cost while preserving pricing accuracy, making SOA particularly suitable for efficiently computing the outcomes for feature vectors in the large-scale datasets in the ML workflow.

Numerical experiments across two option types (European and digital) and three stock price models (GBM, HM, and EVGP) empirically confirm that SOA is substantially more efficient than CMA, typically requiring only $ 60\%-70\% $ of the execution time of CMA under equivalent error tolerances.
The proposed ML-based framework further enhances flexibility and robustness. Three ML algorithms including NN, RF and GBDT are trained on the SOA-generated dataset. Comparative analyses indicate that all three models achieve high accuracy. 
NN achieves the highest accuracy and the greatest improvement in execution time. Although GBDT exhibits slightly lower accuracy and a longer execution time than the NN, it strikes an effective balance among precision, interpretability, and computational efficiency. The selection between GBDT and NN should be guided by the nature of the task and the computational capabilities of the available hardware.
Although an alternative FFT-based framework provides computational acceleration through the FFT and represents a natural extension of SOA (or CMA), it suffers from numerical instability for deep out-of-the-money options and imposes strict requirements on input option attributes.
The ML-based framework overcomes these limitations, making it particularly suitable for practical deployment in financial institutions that price large and heterogeneous option portfolios.
\section*{Acknowledgments}
The authors would like to express their appreciation to the referees for their useful comments and the editors. Liying Zhang is supported by the National Natural Science Foundation of China (No. 11601514 and No. 11971458), the Fundamental Research Funds for the Central Universities (No. 2023ZKPYL02 and No. 2023JCCXLX01) and the Yueqi Youth Scholar Research Funds for the China University of Mining and Technology-Beijing (No. 2020YQLX03). 2025 Basic Sciences Initiative in Mathematics and Physics.


\begin{thebibliography}{100}
\bibitem{1} Ballotta L, Kyriakou I. Monte Carlo simulation of the CGMY process and option pricing[J]. Journal of Futures Markets, 2014, 34(12), 1095–1121.
\bibitem{2} Black F, Scholes M. The pricing of options and corporate liabilities[J]. Journal of Political
Economy, 1973, 81(3), 637–654.
\bibitem{3} Boyle P P. A lattice framework for option pricing with two state variables[J]. Journal of Financial and Quantitative Analysis, 1988, 23(1), 1–12.
\bibitem{4} Brennan M J, Schwartz E S. Finite difference methods and jump processes arising in the pricing of contingent claims: A synthesis[J]. Journal of Financial and Quantitative Analysis, 1978, 13(3), 461–474.
\bibitem{5} Broadie M, Kaya Ö. Exact simulation of stochastic volatility and other affine jump diffusion processes[J]. Operations Research, 2006, 54(2), 217–231.
\bibitem{6} Carr P, Geman H, Madan D B and Yor M.  The fine structure of asset returns: an empirical investigation[J]. The Journal of Business, 2002, 75(2), 305–332.
\bibitem{7} Carr P, Madan D. Option valuation using the fast fourier transform[J]. Journal of Com-putational Finance, 1999, 2(4), 61–73.
\bibitem{8} Cox J C, Ross S A, Rubinstein M. Option pricing: a simplified approach[J]. Journal
of Financial Economics, 1979, 7(3), 229–263.
\bibitem{9} Derman E, Kani I.  The volatility smile and its implied tree[J]. Goldman Sachs Quantitative Strategies Research Notes, 1994, 2, 45–60.
\bibitem{10} Elbrächter D, Grohs P, Jentzen A, Schwab C. DNN expression rate analysis of high-dimensional pdes: application to option pricing[J]. Constructive Approximation, 2022, 55(1), 3–71.
\bibitem{11} Ferguson R, Green A. Deeply learning derivatives. arXiv:1809.02233, 2018.
\bibitem{12} Genlöv Brouwer C. Numerical investigation of the deep bsde solver for pricing European options[J]. 2025.  
\bibitem{13} Heston S L. A closed-form solution for options with stochastic volatility with applications to bond and currency options[J]. The Review of Financial Studies, 1993, 6(2), 327–343.
\bibitem{14} Hirsa A. Computational methods in finance[M]. Chemical Rubber Company Press, 2012.
\bibitem{15} Hutchinson J M, Lo A W, Poggio T. A nonparametric approach to pricing and hedging derivative securities via learning networks[J]. The Journal of Finance, 1994, 49(3), 851–889.
\bibitem{16} Jeon Y, Mccurdy T H, Zhao X. News as sources of jumps in stock returns: Evidence from 21 million news articles for 9000 companies[J]. Journal of Financial Economics, 2022, 145(2), 1–17.
\bibitem{17} Li Z, Huang Q. Option pricing using ensemble learning. arXiv:2506.05799, 2025.
\bibitem{18} Lux T, Marchesi M. Volatility clustering in financial markets: a microsimulation of
interacting agents[J]. International Journal of Theoretical and Applied Finance, 2000, 3(04), 675–702.
\bibitem{19} Madan D B, Carr P P, Chang E C. The variance gamma process and option pricing[J].
Review of Finance, 1998, 2(1), 79–105.
\bibitem{20} Merton R C. Option pricing when underlying stock returns are discontinuous[J]. Journal of
Financial Economics, 1976 ,3(1-2), 125–144.
\bibitem{21} Patel R G. Efficient deep learning methods for solving high-dimensional partial differential equations for applications in option pricing[D]. University of Toronto (Canada), 2022.
\bibitem{22} Poirot J, Tankov P. Monte carlo option pricing for tempered stable (CGMY) processes[J]. Asia-Pacific Financial Markets, 2006, 13(4), 327–344.
\bibitem{23} Rachev S T, Kim Y S. Bianchi, M. L. and Fabozzi, F. J. Financial models with Lévy processes and volatility clustering[M]. John Wiley and Sons, 2011.
\bibitem{24} Rapuch  G. American options and the free boundary exercise region: a PDE approach[J]. Interfaces and Free Boundaries, 2005, 7(1), 79–98.
\bibitem{25}Stein E M, Shakarchi R. Fourier analysis: an introduction[M]. Princeton: Princeton University Press, 2003.
\bibitem{26} Wang X, Li J, Li J. A deep learning based numerical PDE method for option pricing[J]. Computational Economics, 2023, 62(1), 149–164.
\end{thebibliography}
\end{document}